\documentclass[12pt]{iopart}

\usepackage{iopams}  
\usepackage{subcaption}
\usepackage{graphicx}%
\usepackage{multirow}%
\usepackage{mathrsfs}%
\usepackage{amsthm}
\usepackage[title]{appendix}%
\usepackage{xcolor}%
\usepackage{textcomp}%
\usepackage{manyfoot}%
\usepackage{booktabs}%
\usepackage{algpseudocode}%
\usepackage{listings}%
\usepackage{harvard}

\usepackage[ruled,vlined]{algorithm2e}
\SetAlCapSkip{1em}           % Space between algorithm caption and body
\SetAlCapNameFnt{\normalfont} % Caption font style
\DontPrintSemicolon          % Don't print semicolons
\LinesNumbered               % Line numbers if you want them
\SetVlineSkip{0.5em}         % Vertical space around lines
\SetAlgoInsideSkip{0.5em}    % Space between lines

\newtheorem{theorem}{Theorem}%  meant for continuous numbers
\newtheorem{lemma}{Lemma}
\newtheorem{corollary}{Corollary}

\newtheorem{innercustomthm}{Theorem}
\newtheorem{definition}{Definition}%

\newenvironment{customthm}[1]
  {\renewcommand\theinnercustomthm{#1}\innercustomthm}
  {\endinnercustomthm}

\begin{document}

\title[]{The Role of Sequence Information in Minimal Models of Molecular Assembly}

\author{Jeremy Guntoro and Thomas E. Ouldridge}

\address{Department of Bioengineering, Imperial College London, Exhibition Road, London SW7 2AZ, United Kingdom}
\ead{t.ouldridge@imperial.ac.uk}
\vspace{10pt}

\begin{abstract}
Sequence-directed assembly processes -- such as protein folding  -- allow the assembly of a large number of structures with high accuracy from only a small handful of fundamental building blocks.  We aim to explore how efficiently sequence information can be used to direct assembly by studying variants of the temperature-1 abstract tile assembly model (aTAM).  We ask whether, for each variant, their exists a finite set of tile types that can deterministically assemble any shape producible by a given assembly model; we call such tile type sets ``universal assembly kits''. Our first model, which we call the ``backboned aTAM", generates backbone-assisted assembly by forcing tiles to be added to lattice positions neighbouring the immediately preceding tile, using a predetermined sequence of tile types. We demonstrate the existence of universal assembly kit for the backboned aTAM, and show that the existence of this set is maintained even under stringent restrictions to the rules of assembly. We compare these results to a less constrained model that we call sequenced aTAM, which also uses a predetermined sequence of tiles, but does not constrain a tile to neighbour the immediately preceding tiles. We prove that this model has no universal assembly kit in the stringent case. The lack of such a kit is surprising, given that the number of tile sequences of length $N$ scales faster than both the number and worst-case Kolmogorov complexity of producible shapes of size $N$ for a sufficiently large -- but finite -- set of tiles. Our results demonstrate the importance of physical mechanisms, and specifically geometric constraints, in facilitating efficient use of the information in molecular programs for structure assembly. 
\end{abstract}

%\maketitle

\section{Introduction}\label{sec1}

\subsection{Motivation and Aims}

Biological systems apply a wide range of molecular assembly maps that accept some generalised ``genotype'' as input and produce some phenotpye as output -- for example, the RNA genotype-phenotype map \cite{Dingle2015} --
 to achieve their biochemical complexity. One type of assembly map involves multiple subunits in solution spontaneously coming together to create a more complex structure, with that structure exclusively determined by the interactions between the subunits; here, the input genotype is an unordered set of building blocks. The archetypal example is the joining of protein subunits to produce a larger quaternary structure \cite{Greenbury2014}. Contrast this map type with sequence-directed assembly processes where subunits are first assembled into a sequence connected by a backbone, before folding into the final structure; here, the input genotype is a {\em sequence} of building blocs. A quintessential example is the folding of polypeptide sequences into protein secondary and tertiary structures, with the sequence specified by an mRNA template \cite{Anfinsen1972}. 

Sequence-directed assembly comes with some unique advantages. For example, it tends to utilize a relatively small set of building blocks to create a vast array of specific structures through sequence variation. This advantage is evident in protein folding, where just 20 amino acid types can produce countless distinct functional structures \cite{Poulton2021,CabelloGarcia2023,Guntoro2025}. Conversely, non-sequenced assembly typically requires a larger set of building blocks to achieve similar structural specificity \cite{Rothemund2000}. If we had a set of 20 building blocks that spontaneously formed a specific structure, they would struggle to form alternative structures containing the same basic subunits \cite{Sartori2020}. While this argument is intuitive, exactly which features facilitate the efficient use of building blocks for a given sequence-directed assembly map is less well understood.

The importance of sequence-directed assembly in biology is appreciated, and a number of models have been put forward to investigate different aspects of sequence-directed assembly. Early examples include 1D lattice models, such as the Lifson-Roig model \cite{Lifson1961}, which were used to characterize helix-coild secondary structure phase transitions. The HP-Lattice model \cite{Lau1989}, an approximate model of protein folding, has been used to understand the role of hydrophobic and hydrophilic amino acids in protein folding. RNA secondary structure folding has been used to evaluate the genotype-phenotype map properties of sequence-directed assembly maps due to the speed with which low free-energy configurations can be identified \cite{Dingle2015}. The computational capabilities of backbone-directed assembly have been explored using the Oritatami model \cite{Geary2018,Han2018,Demaine2018}. 

One important motivation for investigating these theoretical models of folding is to further develop artificial molecular assembly systems, including DNA nanotechnology-based approaches \cite{Rothemund2006,Seeman2017}. Notably, while DNA itself is a copolymer of backbone-linked nucleotides, DNA assembly systems utilize both sequence-directed and sequence-free assembly. Some assembly systems rely primarily on the free assembly of short oligonucleotide strands \cite{Videbaek2022,Ke2012,Mohammed2013}. DNA origami \cite{Rothemund2006} represents an interesting hybrid approach, using a long scaffold strand folded by short staples that bind to non-contiguous domains; systems with reusable or reconfigurable modules attempt to more heavily exploit features of sequence-directed assembly \cite{Young2020,Dunn2015}. Single-stranded nucleic acid nanotechnology \cite{Shih2004,Geary2014,Zhou2020,Kocar2016,Han2017} and alternative techniques such as programmable droplets \cite{McMullen2022} take greater advantage of sequence-directed assembly. 

The models we explore in this work are based on the polyomino assembly model applied in \cite{Ahnert2010}, itself based on the temperature-1 abstract Tile Assembly Model (aTAM) \cite{Rothemund2000} allowing for tile rotations \cite{Demaine2012}. aTAM variants are frequently employed to explore the computational capabilities of self-assembly systems. In particular, much of the existing literature on the aTAM has focused on the question of intrinsic universality,  the ability of an assembly system to simulate any other instance of the assembly scheme at some scaling factor \cite{Woods2015}. The temperature-1 aTAM is not intrinsically universal \cite{Doty2011,Meunier2017,Meunier2020}, while the temperature-2 aTAM is \cite{Doty2010,Demaine2012,Woods2015}. Other work has investigated the minimum number of tiles required to produce specific shapes \cite{Soloveichik2007,Meunier2020,Patitz2011} from some starting seed and at some scale factor. Yet another line of literature considers minimal tilesets without starting seeds or scale factors \cite{Ahnert2010,Greenbury2014,Johnston2022}; these works consider the physical properties of the resulting maps from the ``genotype'' of tile types to the assembly phenotype, and their relationship to the Kolmogorov complexities of polyominoes. This line of work has additionally resulted in Boolean Satisfiability algorithms for tile minimization \cite{Russo2022,Bohlin2023}.  

In this paper, we compare different sequence-directed assembly models and consider their ability to deterministically assemble complex shapes using only a finite set of tiles; in doing so, we hope to develop an understanding of the features of efficient sequence-directed assembly maps. Our choice of the aTAM as a base model is motivated by this aim, as the underlying rules of the aTAM are simple enough that we can construct toy models that isolate specific features of real-world, sequence-directed assembly maps while selectively omitting others. For these toy models of assembly, we consider whether the information in a sequence can be used to direct the assembly of any possible shape, and whether this information can be used efficiently. We consider two measures of efficiency, one related to the scaling of the number of possible shapes with size, and the other related to the scaling in Kolmogorov complexity of shapes with size. Kolmogorov complexity is the length of a minimal program (under some universal language) required to produce a certain output \cite{Kolmogorov1968}. While the Kolmogorov complexity is formally incomputable, progress can be made by considering its upper bounds and scaling \cite{Ahnert2010,Greenbury2014,Johnston2022}. 

This paper is organized as follows. In Section \ref{sec:model_def}, we present definitions of models and terms we use. In Section \ref{sec:results}, we present our results (these are briefly summarized in Section \ref{sec:summary}) and discuss these results in light of the properties of our models. We discuss the conclusions and implications of this work in Section \ref{sec:conclusion}.

\subsection{Summary of Results}
\label{sec:summary}

We first consider an aTAM-based model that simulates a backbone-assisted assembly process, which we call the backboned aTAM. Like the base aTAM, the fundamental building blocks of our system are tiles with labeled faces. However, tiles can only be added in a specified order of tile types, and added tiles must neighbour the last added tile. The model can be thought of as mimicking idealized co-transcriptional folding; the added sequence information present in the backboned aTAM means it can be more powerful at selectively producing specific shapes, allowing for a  more repetitive use of tiles.

We provide here an informal definition of a \textbf{universal assembly kit} as this concept is vital to summarizing our results (a more thorough definition is presented in Section \ref{sec:backboned_proof}). Let an assembly scheme be a set of rules that constructs shapes from tiles with numbered faces and possible additional inputs (in our case, these additional inputs are usually a sequence in which tiles can be added). Then, a universal assembly kit is a finite set of tile types and interaction rules that allow any finite shape to be constructed deterministically, if that shape can be constructed at all by the assembly scheme (i.e. without imposing constraints on the set of tile types that may be used, some input to the assembly scheme exists that constructs the shape). 

Our first result is the existence of a universal assembly kit for the backboned aTAM.
\begin{customthm}{1}
A universal assembly kit for the backboned aTAM exists.
\end{customthm} 
The construction we arrive at to prove Theorem {\ref{theorem:walk_main}} is in some sense artificial as it relies exclusively on backbone routing to define shapes while disregarding non-backbone interactions; interactions between tiles that are not connected by the backbone are all ``neutral''. To circumvent this problem, and get closer to protein folding in which interactions with non-neighbouring amino acids are essential in determining the fold, we show that the existence of a universal assembly kit is preserved even when neutral interactions are forbidden and all adjacent tile faces are required to have attractive glue interactions. This finding recalls existing work on the aTAM, where mismatches and rotations were shown to have a weak effect on computational power \cite{Mavnuch2009}.

\begin{customthm}{2}
A universal assembly kit for the backboned aTAM exists such that all inter-tile interactions in any configuration are required to be attractive. 
\end{customthm}

To help us understand the role of sequence information in very general cases, we develop a model that we call the sequenced aTAM, where tiles are added in a predetermined sequence, as in the backboned aTAM, but added tiles are no longer constrained to neighbor preceding tiles. While not reflective of an {\it autonomous} real-world process, such a model could describe a generalisation of the step-by-step synthesis protocols that are used to construct synthetic DNA and proteins \cite{Beaucage1981,Dawson1994}. For our purposes, this construct serves as comparative model in which sequence information is not paired with geometric constraints. We find that unlike the backboned aTAM, a universal assembly kit does not exist in the case of the sequenced aTAM.

\begin{customthm}{4}
The sequenced aTAM does not admit a universal assembly kit.
\end{customthm}

These results on assembly kits can be interpreted in light of the scaling of shape space size and worst-case Kolmogorov complexity of the underlying shape space. Neither the shape space size nor the Kolmogorov complexity prohibit a universal assembly kit for either the backboned or sequenced aTAM. Indeed, the shape space provides a relatively small lower bound on the number of tile types. We find that, although the backboned aTAM does possess a universal assembly kit, requiring that all interactions are attractive makes the resultant ``program'' rather inefficient, due to both redundancy and failed assemblies. These inefficiencies are so large for the sequenced aTAM that no universal assembly kit is possible at all. This result hints at a role of the backbone in reducing assembly complexity beyond coupling shape assembly to a sequence. 

Note, however, that the absence of backbone restrictions means that the sequenced aTAM can assemble shapes that are not formed by a single self-avoiding walk, whereas the backboned aTAM cannot, and therefore the sequenced aTAM has a larger structure space. We cannot therefore conclude that the backbone exclusively confers advantages to assembly. Moreover, the sequenced aTAM can, with a finite asssembly kit, deterministically assemble any shape in the structure space of the backboned aTAM, provided that neutral interactions are allowed.  However, we have been unable to prove whether the sequenced aTAM can -- like the backboned aTAM -- do so even when we impose attractive inter-tile interactions. We conclude by presenting some partial results in this direction. 

\section{Model Definitions}
\label{sec:model_def}

We begin with definitions borrowed from the aTAM literature \cite{Rothemund2000,Doty2011}.  As tiles are allowed to rotate, we differentiate between an oriented tile, which is  a 4-tuple of glue types $(\sigma_N, \sigma_E, \sigma_S, \sigma_W)$, and the orientation-free tile type which is an equivalence class of all cyclic permutations of any of its given tiles. In practice, the tile type can be associated with a default orientation, and the tile can be conceived of a tile type $\Theta$ placed in some specific orientation $\rho = \{N,E,S,W\}$ corresponding to the direction faced by the face that would face north in the default orientation. 

\begin{figure}
    \centering
    \includegraphics[width=0.8\textwidth]{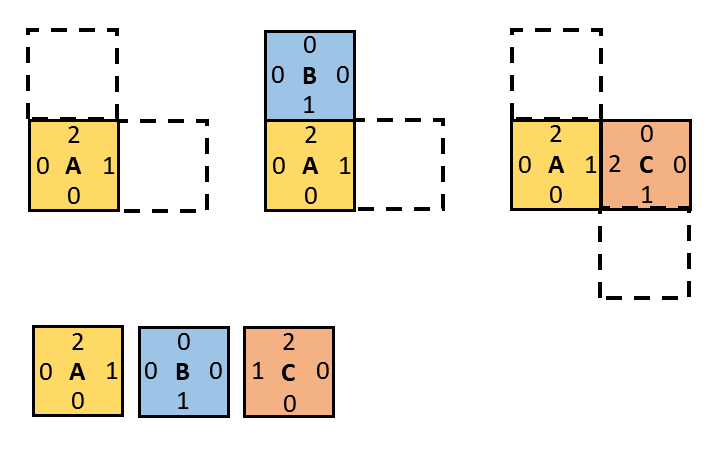}
    \caption{ A diagram illustrating aTAM assembly. The fundamental building blocks of assembly are square tiles with numbered faces. At each step, a tile drawing from the assigned set of tile types \textbf{(bottom)} is added to a random position neighbouring an existing tile, with the possible locations being restricted to those where the resulting sum of strengths of interactions formed is greater than or equal 1 (for temperature-1 aTAM). Here, the glue pair $(1,2)$ are predetermined to have strength of 1. Note that we use a variant of the aTAM in which tiles can be rotated. Example assembly steps steps starting from the state at the \textbf{top left} of the figure are given in \textbf{top centre} and \textbf{top right}.  }
    \label{fig:polyomino_assembly}
\end{figure}

A \textbf{configuration} is a partial function $A:\mathbb{Z}^2 \rightarrow \mathbb{T}$, where $\mathbb{T}$ is the set of all possible tiles. $dom(A)$ is the set of points in configuration $A$ with a tile. A coordinate $z \notin dom(A)$ is empty in $A$. $A$ is a \textbf{subconfiguration} of $A'$ if $dom(A) \subseteq dom(A')$. For convenience, we frequently use single-tile configurations $a = (\Theta, \rho, z)$ for a tile type $\Theta$, an orientation $\rho$ and a coordinate $z$, which we call \textbf{coordinated tiles}. The empty configuration is defined as $A_{empty}$ such that $dom(A_{empty}) = \emptyset$. The addition of configurations $A'' = A + A'$ is well-defined if $dom(A) \bigcap dom(A') = \emptyset$, otherwise $A'' = \infty$ \cite{Rothemund2000}. In the former case, 

 \[
        A''(z) = 
        \begin{cases}
            A(z) \;\: if \;\: z \in dom(A),\\ 
            A'(z) \;\: if \;\: z \in dom(A').\\
        \end{cases}
\]

The \textbf{strength function} is a partial function $g:\mathbb{Z}_+^2 \rightarrow \mathbb{Z}$. The strength function determines the type of interaction between two glues. Two glues $\sigma$ and $\sigma'$ have an attractive interaction if $g(\sigma,\sigma') > 0$, do not interact if $g(\sigma,\sigma') = 0$ (also called a neutral interaction) and have a repulsive interaction if $g(\sigma,\sigma') < 0$. Newly added tiles are allowed to form neutral or repulsive interactions (see  \cite{Mavnuch2009} for a further discussion of negative strengths), but only if the sum of interactions of each of their edges is equal to or greater than the temperature $\tau$, which we have set equal to $1$ for all the models that we consider. 

The aim of any instance of an aTAM system is to produce shapes. For the base aTAM, shapes are assembled through the addition of tiles drawing from the tile type set $T$ until no further tiles can be added (Refer to Figure \ref{fig:polyomino_assembly} for an illustration, and reference \cite{Rothemund2000} for formal definitions). We describe below the way in which our two models, the \textbf{backboned aTAM} and the \textbf{sequenced aTAM}, assemble their shapes. 

We refer to the set of non-empty points $dom(A)$ as the \textbf{oriented shape} of a configuration $A$. We consider all outputs that can be transformed through rotations and translations as equivalent; hence we rely on a more general notion than oriented shape. A \textbf{shape} is thus an equivalence class of oriented shapes, containing oriented shapes that can be transformed to each other via rotations and translations, and the shape of a configuration is the shape to which its \textbf{oriented shape} belongs. We denote by $\simeq$ the equivalence relation that defines a shape, such that if two oriented shapes $S$ and $S'$ have the same shape, then $S \simeq S'$. Hence, the only ambiguity in shape is overall rotation or translation, which do not violate the equivalence class.

\begin{figure}
    \centering   \includegraphics[width=0.8\textwidth]{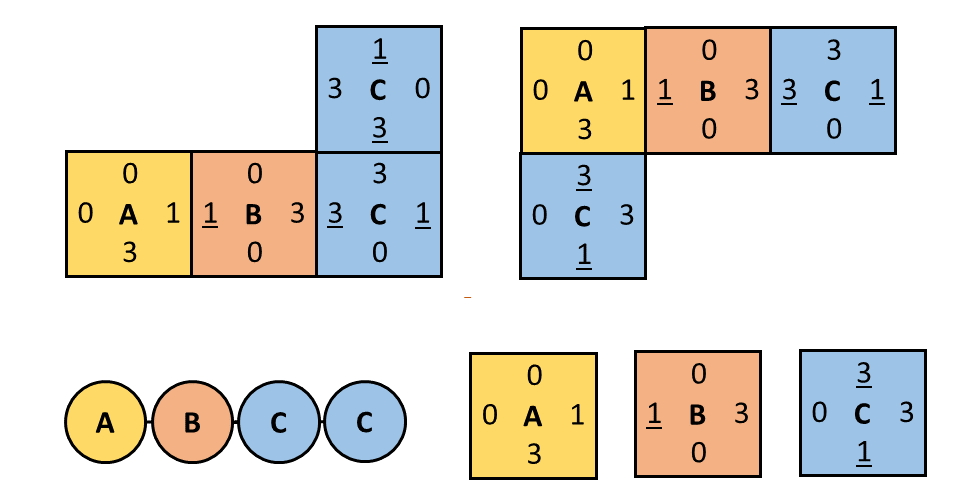}
    \caption{A figure illustrating assembly via the backboned aTAM and the sequenced aTAM. Consider an instance $(A_{empty},Q,g)$ of either the backboned aTAM or the sequenced aTAM, and where $g(x,\underline{x}) = 1$. An example sequence $Q$ is provided in the \textbf{bottom left}, with the letters of each tile type corresponding to tiles found in the \textbf{bottom right}. For a backboned aTAM instance, added tiles must neighbour the tile added in the immediately preceding step, and hence only the \textbf{top left} configuration can result from the backboned aTAM. By contrast, the sequenced aTAM has no such restriction, and both top left and top right configurations can be the final configuration in a trajectory of a sequenced aTAM instance.}
    \label{fig:backboned_sequenced}
\end{figure}

Intuitively, the \textbf{backboned aTAM} aims to mimic a cotranslational folding system. Unlike the base aTAM, tiles can only be added in a predetermined order (a sequence of tile types is provided as an input to the assembly system), and any added tiles must neighbour the last added tile (Figure \ref{fig:backboned_sequenced}).  Formally, a \textbf{backboned aTAM instance} is a 3-tuple $(\mathcal{A}, Q, g)$, where $\mathcal{A}$ is an initial configuration, $Q$ is an input sequence of tile types and $g$ is a strength function, that generates a set of complete trajectories, with trajectories being defined in Definition 1, and with the set of complete trajectories obeying Definition \ref{def:backbone_complete}. Note that in our treatment of the interaction function $g$, we have allowed negative (respulsive) interactions. However, unlike the approach taken by \cite{Doty2011b} and \cite{Patitz2011}, we do not allow negative interactions to displace existing tiles - rather, tiles that would form strong enough negative interactions to destabilize the configuration are simply blocked.

Note that outside of specific theoretical constructs, $\mathcal{A}$ is usually taken as the empty configuration, so that the backboned aTAM is normally conceived as a seedless assembly system. Compared to the base aTAM, the backboned aTAM accepts a sequence of tile types as opposed to a set of tile types as input. For the base aTAM, an operator $\rightarrow^*_{\mathbf{T}}$ was used to define the assembly of shapes. However, this is not an ideal descriptor for the backboned aTAM as it would obscure the contribution of the sequence. Rather, we build our definition from trajectories, which we define as follows:

\begin{definition}
A \textbf{trajectory} $\Psi = (A_0,A_1,A_2,..)$ is a sequence whose elements are either configurations or $\infty$.
\end{definition}
$\infty$ is generally permitted only in the context of describing forbidden trajectories (that lead to two overlapping tiles). 

A specific instance of the backboned aTAM can generate trajectories through the addition of tiles with types given by the sequence $Q$ of the backboned aTAM system, while obeying the strength function $g$ and rules about tile placement. A trajectory that can be generated by a backboned aTAM system is said to be complete with respect to the backboned aTAM system. We can now formally define a complete trajectory. Note that in using cuts and cut-edges, we follow the approach in \cite{Doty2011b} and \cite{Patitz2011}.

\begin{definition}
Consider a backboned aTAM instance $(\mathcal{A},Q,g)$, with the sequence of tile types $Q = (\Theta_1, \Theta_2,.., \Theta_\zeta)$. A trajectory $\Psi = (A_0,A_1,A_2,..)$ is said to be complete with respect to $(\mathcal{A},Q,g)$ if the following hold. 
\begin{enumerate}
\item Starting configurations are consistent, that is $A_0 = \mathcal{A}$.
\item $A_{t} = A_{t-1} + a_t$ with $a_t$ chosen randomly under the constraints:
\begin{enumerate}
    \item $A_t \neq \infty$ for any $t$.
    \item $a_t$ is the coordinated tile formed by $\Theta_t$ after undergoing some rotation and placed on a coordinate $z$. 
    \item Define $G_A(A)$ of a configuration $A$ to be a graph whose vertices are non-empty coordinated tiles and edges are drawn between neighbouring coordinated tiles. Let a \textbf{cut} be defined in the usual graph-theoretic way, as a partition of graph vertices into two disjoint subsets. The set of \textbf{cut-edges} is then the set of edges with one vertex in each partition. Then, any cut of $G_A(A_t)$ is such that the sum of interaction strengths over cut-edges $\sum_i g(\sigma_i,\sigma_i') > 0$, where $\sigma_i$ and $\sigma_i'$ are neighbouring faces corresponding to an edge. A configuration $A_t$ that obeys this assumption for some interaction function $g$ is called \textbf{g-valid}. Informally, a g-valid configuration is one in which any subconfiguration of tiles in a configuration will be bound to the rest of the configuration with a total strength of at least 1.
    \item  Each added coordinated tile $a_t$ forms an attractive interaction with the coordinated tile $a_{t-1}$ added in the last time step unless $t = 0$ or $A_{t-1}$ is empty.
\end{enumerate}
\item The trajectory terminates upon reaching $A_\zeta$ or when no such single tile configuration $a_t$ can be added. In the latter case, the trajectory is said to have been prematurely terminated.
\end{enumerate}
\label{def:backbone_complete}
\end{definition}

By this definition, the first tile added in an empty configuration can be placed at any coordinate, in any orientation. The set of assembled configurations for an instance of the backboned aTAM is defined as the set containing all terminal configurations for all complete trajectories generated by the backboned aTAM instance. The set of assembled shapes for a backboned aTAM instance is thus the set containing all shapes of assembled configurations of the backboned aTAM instance. 

We further define a model that we call the sequenced aTAM, which retains the sequence of the backboned aTAM, but not the backbone constraints. The definition of a sequenced aTAM instance, along with the definition of its complete trajectories, are as follows.  A \textbf{sequenced aTAM instance} is a 3-tuple $(\mathcal{A}, Q, g)$ (with $\mathcal{A}, Q, g$ defined identically to a backboned aTAM instance) that generates a set of complete trajectories obeying Definition \ref{def:sequence_complete}, rather than Definition \ref{def:backbone_complete}. 

\begin{definition}
Consider a sequenced aTAM instance $(\mathcal{A},Q,g)$, with the sequence of tile types $Q = (\Theta_1, \Theta_2,.., \Theta_\zeta)$. A trajectory $\Psi = (A_0,A_1,A_2,..)$ is said to be complete with respect to $(\mathcal{A},Q,g)$ if the following hold. 
\begin{enumerate}
\item Starting configurations are consistent, that is $A_0 = \mathcal{A}$.
\item $A_{t} = A_{t-1} + a_t$ with $a_t$ chosen randomly under the constraints:
\begin{enumerate}
    \item $A_t \neq \infty$ for any $t$.
    \item $a_t$ is the coordinated tile formed by $\Theta_t$ after undergoing some rotation and placed on a coordinate $z$. 
    \item Any cut of $G_A(A_t)$ is such that the sum of interaction strengths over cut-edges $\sum_i g(\sigma_i,\sigma_i') > 0$, where $\sigma_i$ and $\sigma_i'$ are neighbouring faces corresponding to an edge. A configuration $A_t$ that obeys this assumption for some interaction function $g$ is called \textbf{g-valid}.
\end{enumerate}
\item The trajectory terminates upon reaching $A_\zeta$ or when no such single tile configuration $a_t$ can be added. In the latter case, the trajectory is said to have been prematurely terminated.
\end{enumerate}
\label{def:sequence_complete}
\end{definition}

The key difference in the definition of a complete trajectory for the sequenced aTAM relative to the backboned aTAM is that it does not mandate that new tiles are placed adjacent to the previous tile. This difference allows for the same $\mathcal{A}, Q, g$ to generate a distinct set of trajectories

We end this section by defining a few ideas necessary to build towards a notion of a universal assembly kit. The set of all shapes that are defined by terminal configurations in any instance of an assembly model (the backboned or sequenced aTAM) is called the \textbf{shape space} of assembly model; informally, it is the set of shapes that can be assembled by the model, allowing any sequence and strength function. A given instance of the backboned or sequenced aTAM, with a specific sequence and strength function, is called \textbf{deterministic} if and only if its set of assembled shapes contains exactly one element. We use \textbf{oriented determinism} to refer to the stronger condition of a system assembling only a single oriented shape (the latter is only possible if the starting configuration is non-empty).

\section{Results}
\label{sec:results}

To address the question of whether various aTAM models can efficiently encode the deterministic assembly of shapes into their sequences, we consider the existence of universal assembly kits, first for the backboned aTAM, and eventually for the sequenced aTAM. A \textbf{universal assembly kit} is defined as a finite tile type set and corresponding strength function that allows the deterministic assembly of any finite shape  within the shape space of the model.  Here, the starting configuration must be the empty set and the sequence of tiles can be arbitrary, within those sequences permitted by the finite tile set.

\subsection{Universal Assembly Kits for Backboned aTAM}
\label{sec:backboned_proof}
 We begin our results with the following Lemma.

%(let such a set of shapes be called the \textbf{shape space} of the model). We call such a tile set a \textbf{universal assembly kit}. Such a tile type set would also have associated with it some specific strength function. \textbf{Deterministic} here refers to a system assembling only a single shape. We use \textbf{oriented determinism} to refer to the stronger condition of a system assembling only a single oriented shape (the latter is only possible if the starting configuration is non-empty).  

\begin{lemma}
The shape space of the backboned aTAM is a subset of the shapes of self-avoiding paths. 
\end{lemma}
\begin{proof}

This lemma trivially proceeds from definition $\ref{def:backbone_complete}$, as added tiles must neighbour the last added tile.
\end{proof}

\begin{theorem}
A universal assembly kit for the backboned aTAM exists.
\label{theorem:walk_main}
\end{theorem}

\begin{proof} Since any shape that can be assembled by a backboned aTAM system is a shape of some self-avoiding walk, we can always assemble any shape that is achievable by a backboned aTAM system by encoding the set of left, right or forward moves of the underlying self-avoiding walk. We use glues $0, 1$ and $2$, along with the following strength function:
\begin{equation}
    g =
    \left\{
    \begin{array}{ll}
      1, & \mbox{if}\ (\sigma,\sigma') = (1,2) \mbox{ or } (\sigma,\sigma') = (2,1),\\
      0, & \mbox{otherwise}.
    \end{array}
    \right.
\label{eq:g_1}
\end{equation}
Then, each left, right or forward move can  be performed by a specific  ``directed'' tile, shown in Figure \ref{fig:blocks_walk}. An additional tile encoding the start of the shape is required, but there is no need for an end tile as any of the directed tiles can be placed at the end of any self-avoiding walk without impacting the final shape. Hence, a universal assembly kit with 4 tiles is sufficient to assemble any shape achievable by the backboned aTAM, without requiring a starting seed configuration or needing to scale up the shape.
\end{proof}

Given this theorem, the following corollary trivially holds. 

\begin{corollary}
The shape space of the backboned aTAM is equal to the set of shapes of self-avoiding paths.
\label{corollary:shape_space}
\end{corollary} 
\begin{proof}
Any self-avoiding path can be constructed using the method in Theorem \ref{theorem:walk_main}.
\end{proof}

\begin{figure}
\centering
\includegraphics[width=0.7\textwidth]{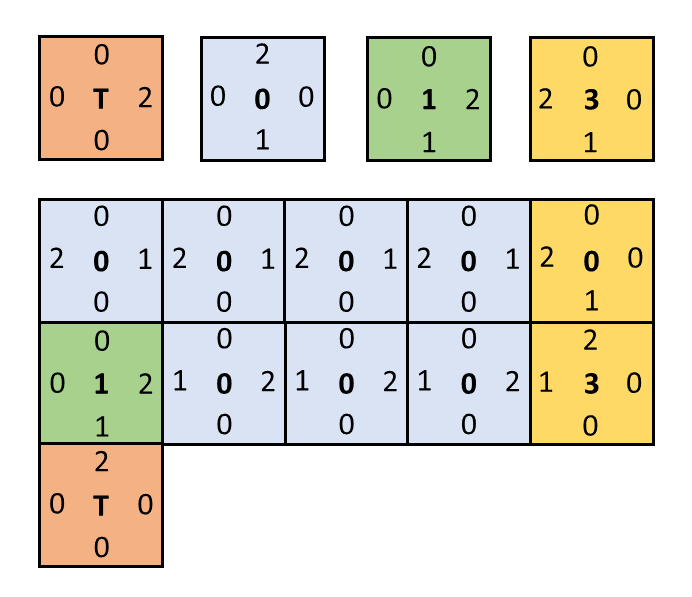}
    \caption{ A figure illustrating a finite set of ``directed'' tiles that comprise a universal assembly kit of the backboned aTAM (\textbf{top}), as well as an example configuration utilizing these tiles (\textbf{bottom}). }  
    \label{fig:blocks_walk}
\end{figure}

\begin{figure}
    \centering
    \includegraphics[width=0.7\textwidth]{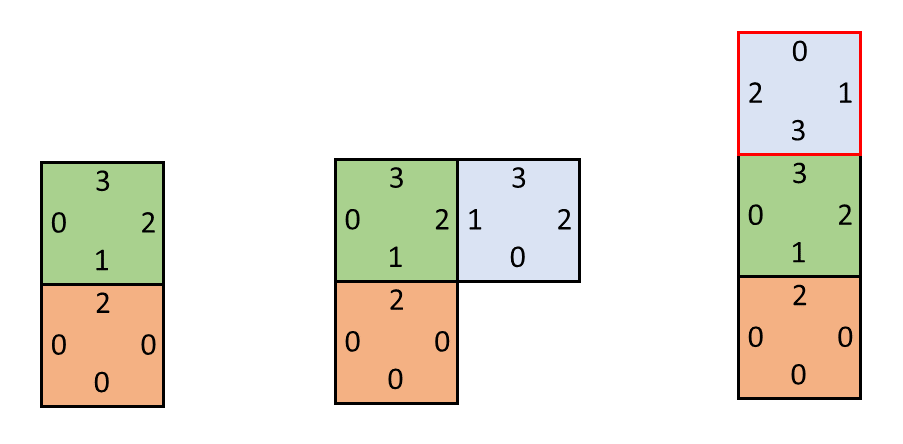}
    \caption{A figure illustrating the difficulties associated with assembly with only attractive (i.e. without neutral) inter-tile interactions. While in principle one can replace the neutral interface type $0$ with a  self-attractive interface type $3$, such attractive interfaces can pull tiles towards unintended positions (\textbf{right}).} 
    \label{fig:incorrect_placement}
\end{figure}

These initial results are fairly straightforward. However, the construct we have arrived at is somewhat artificial when considering the physical system being emulated (co-translational folding), since the shape is completely determined by the backbone routing and non-backbone-adjacent interactions are irrelevant. This result is only possible due to the use of neutral interactions, or by allowing adjacent tiles to have weak repulsive interactions (that are not sufficiently strong to prevent the addition of a tile). We therefore now consider the backboned aTAM with an added restriction that all inter-tile interactions in any configuration are attractive. There is a sense in which this variant is closer to biological folding maps, where matching non-backbone interactions are necessary for accurate folding. One strategy for devising a universal assembly kit is to use the left-right-forward tile types in \ref{fig:blocks_walk}, but to replace the $0$ glues with custom glues depending on interactions. Immediately from Figure \ref{fig:incorrect_placement}, more glue types, and hence tile types, would be required to avoid incorrect assembly. We now proceed to identify a set of glues and tile types that can assemble all possible shapes, thereby proving the existence of a universal assembly kit. 

\begin{theorem}
A universal assembly kit for the backboned aTAM exists such that all inter-tile interactions in any configuration are required to be attractive. 
\label{theorem:interacting_walk}
\end{theorem}

\begin{proof} The existence of some infinite tile set is trivial, as in the construction of all possible shapes, one can use the tile set in Theorem \ref{theorem:walk_main}, but replace all glue types of all faces that neighbour another tile face with some unique attractive glue type. However,a universal assembly kit requires a finite number of tiles and reusing the same attractive glue type multiples times within a configuration can result in non-deterministic assembly as tiles can be attracted into incorrect positions (Figure \ref{fig:incorrect_placement}). Hence, we must proceed to construct a scheme for numbering tile types such that only a finite number of glue types $N$ are required while still guaranteeing deterministic shape assembly. Let the glue types be labeled $\{0,1,2,3,4,5...N-1\}$, and let the strength function $g$ be defined as follows. 
\begin{equation}
    g(\sigma,\sigma')=
    \left\{
    \begin{array}{ll}
      1, & \mbox{if}\ (\sigma,\sigma') = (1,2) \mbox{ or } (\sigma,\sigma') = (2,1) \mbox{ or } ( \sigma = \sigma' \mbox{ and } \sigma > 2 ),\\
      -3, & \mbox{otherwise}.
    \end{array}
    \right.
\label{eq:g_2}
\end{equation}
For this $g(\sigma,\sigma')$, interactions are either attractive or so repulsive that they would preclude tile placement. In any configuration formed by this tile set, all interactions must therefore be attractive, as required.

We call the fundamental tile types for our construction ``interacting directed tile types'',  and define them as follows.
\begin{definition}
\label{def:tiles}
    Consider the following set of tile types:
    \begin{enumerate}
    \item $\hat{\Theta}_T(\sigma, \sigma', \sigma'') = ( 2,\sigma, \sigma', \sigma'' )$, $\sigma, \sigma', \sigma'' = 0,3$;
    \item $\hat{\Theta}_0(\sigma, \sigma') = ( 2,\sigma,1,\sigma' )$, $\sigma, \sigma'=0,3,4,5,...N-1$;
    \item $\hat{\Theta}_1(\sigma, \sigma') = ( 2,1,\sigma,\sigma' )$, $\sigma, \sigma'=0,3,4,5,...N-1$; 
    \item $\hat{\Theta}_3(\sigma, \sigma') = ( 2,\sigma,\sigma',1 )$, $\sigma, \sigma'=0,3,4,5,...N-1$;
    \item $\hat{\Theta}_H(\sigma, \sigma', \sigma'') = ( 1, \sigma, \sigma', \sigma'' )$, $\sigma, \sigma', \sigma''=0,3,4,5,...N-1$.
    \end{enumerate}
    We describe these tile types as \textbf{interacting directed tile types} (Figure \ref{fig:interacting_directed}).
\end{definition} 
     These tile types are constructed by analogy with the tile types used in the proof of Theorem~\ref{theorem:walk_main}. For example, $\hat{\Theta}_0$ is a set of tile types with glue types 1 and 2 (\textbf{backbone glue types}) in the same pattern as tile type 0 in Fig.~\ref{fig:blocks_walk}, but with variable glues on the other faces; $\hat{\Theta}_0(\sigma, \sigma')$ identifies a specific member of that set. 
     Non-backbone glue types are divided into zero and non-zero types; we denote by $B^c({\Theta})$ the set of non-zero, non-backbone glue types on tile type ${\Theta}$. 
    Tile faces endowed with backbone and non-backbone glue types are correspondingly called \textbf{backbone faces} and \textbf{non-backbone faces}, respectively.

\begin{figure}
    \centering
    \includegraphics[width=0.7\textwidth]{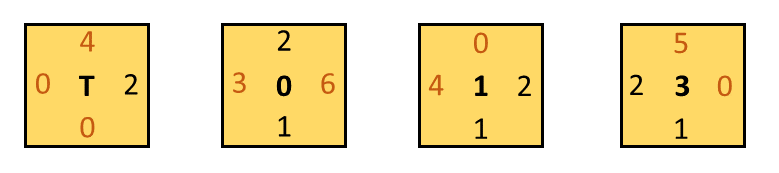}
    \caption{ Example interacting directed tiles, with red/brown glue types representing $\sigma$, $\sigma'$ or $\sigma''$ and black glue types representing backbone faces. }
    \label{fig:interacting_directed}
\end{figure}

We present an algorithm for selecting a sequence $Q$, using a subset of the tiles from Definition~\ref{def:tiles} with $N=7$, to define a backboned aTAM system $(\mathcal{A}, Q, g)$ that can deterministically assemble any given shape within the shape space of the backboned aTAM. The interacting directed tile types in Definition~\ref{def:tiles}, along with the strength function in Equation~\ref{eq:g_2}, therefore define a universal assembly kit for the backboned aTAM, proving the existence of such a kit by construction.

Let $\Psi = (A_{\rm empty}, A_1, ... A_\zeta)$ be a complete trajectory that assembles a shape $\overline{S}$, defining an arbitrary Hamiltonian path through the shape. A sequence of backbone faces consistent with this trajectory can then be selected as in the proof of Theorem~\ref{theorem:walk_main}, fixing whether each tile in the sequence is drawn from the set $\hat{\Theta}_T$, $\hat{\Theta}_0$, $\hat{\Theta}_1$, $\hat{\Theta}_3$ or $\hat{\Theta}_H$. Hence, we only need to select the non-backbone interactions of tile types in the sequence in such a way as to ensure deterministic production of the desired shape.

Assume that the subtrajectory $(A_{\rm empty}, A_1, ..., A_{t-1})$ is given, and that we wish to obtain the next coordinated tile $a_t$ such that $A_t = A_{t-1} + a_t$. 
%The coordinates and backbone orientations of all coordinated tiles $a_t$ are imposed by the path. Hence, we only need to select the non-backbone interactions of tile $\Theta_t$ (where $\Theta_t$ is the tile type corresponding to the coordinated tile $a_t$) in such a way as to ensure deterministic production of the desired shape. 
We apply the following rules to select the non-backbone faces of $a_t$ and hence specify $\Theta_t$, the $t$th tile type in the sequence $Q$:
\begin{enumerate}
    \item Non-backbone faces of $a_t$ with a neighbour in $A_t$ are made to match their neighbouring glue type. 
    \item Non-backbone faces of $a_t$ with no neighbour in $A_\zeta$ (the terminal configuration of $\Psi$) are assigned the repulsive glue type 0. 
    \item Non-backbone faces of $a_t$ with no neighbour in $A_t$ but with a neighbour in $A_\zeta$ are assigned an unknown glue type $\sigma_u$
    \item If the coordinated tile $a_t$ has two (or more) unknown glue types $\sigma_u$ and $\sigma_u'$, then set $\sigma_u = \sigma_u'$
\end{enumerate}

\begin{figure}
    \centering
    \includegraphics[width=0.6\textwidth]{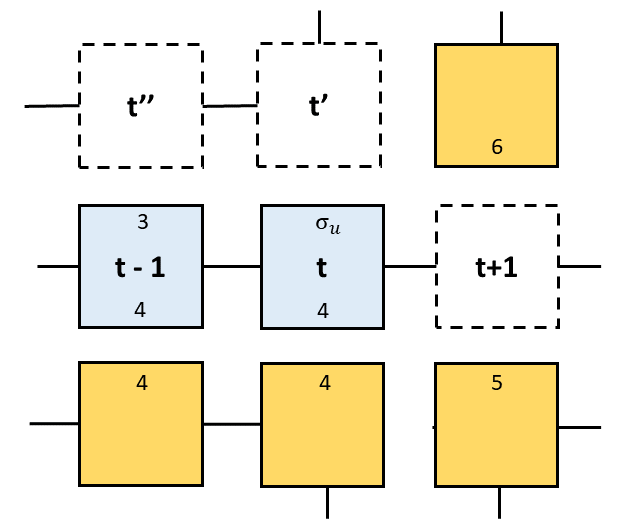}
    \caption{An illustration of the inequalities on the faces of tiles when constructing an assembly kit for the backboned aTAM with only attractive inter-tile interactions. After setting the glue types of known faces, there remains one face with unknown glue type $\sigma_u$. $\sigma_u \neq 3$, $\sigma_u \neq 5$ and $\sigma_u \neq 6$, so $\sigma_u = 4$ is the only correct option if we restrict ourselves to $4$ attractive non-backbone glue types.     }
    \label{fig:inequalities}
\end{figure}

There are two sets of inequalities that must be fulfilled by $\sigma_u$ (see example in Figure \ref{fig:inequalities}):
\begin{enumerate}
\item Define by $\tilde{B^c}(a_{t-1}, A_{t-1}) \subseteq B^c(\Theta_{t-1})$ the set of non-backbone, non-zero glue types of $a_{t-1}$ assigned to at least one face of $a_{t-1}$ with an empty adjacent position in $A_{t-1}$. Then, $\sigma_u \not\in \tilde{B^c}(a_{t-1}, A_{t-1})$ to stop a tile with type $\Theta_t$ binding to $a_{t-1}$ in the incorrect location. 
\item Define by $M(a_{t+1}, A_{t+1})$ the set of non-backbone, non-zero glue types of faces adjacent to $a_{t+1}$ in $A_{t+1}$.
Then, $\sigma_u \not\in M(a_{t+1}, A_{t+1})$ to stop a tile with type $\Theta_{t+1}$ binding to $a_{t}$ incorrectly. 
\end{enumerate}

Our construction means that $\tilde{B^c}(\Theta_{t-1})$ has at most $1$ member, while lattice placement rules mean that $M(a_{t+1}, A_{t+1})$ has at most $2$ members, since the next tile can only have two non-backbone-connected neighbours. Hence, there are at most 3 inequalities on the sole unassigned glue type $\sigma_u$, and hence $4$ attractive glue types are always sufficient to ensure that $\sigma_u$ has an assignment that allows deterministic production of the desired shape. A worst-case scenario is illustrated in Figure \ref{fig:inequalities}.

The arguments above break down for the penultimate tile as the final tile can have three non-backbone-connected neighbours. However, in this case all of the final neighbours of $a_\zeta$ and $a_{\zeta-1}$ are already in place. Therefore the non-backbone glue types of 
$a_{\zeta-1}$ are either specified by neighbouring tiles that are present in the configuration $A_{\zeta-1}$, or 0. They cannot, therefore, cause any ambiguity when $a_\zeta$ is placed, with glue types set either to match neighbours present in configuration $A_{\zeta}$, or 0. Taking $N=7$, 208 tiles are defined by Definition~\ref{def:tiles}. This number provides an upper bound on the size of the minimal universal assembly kit for the backboned aTAM. 

\end{proof}

%we note that there can be no ambiguous interfaces $\sigma_u$ as lattice placement rules prevent $a_\zeta$, the only remaining unadded tile, from neighbouring $a_{\zeta-2}$. Similarly, there can be no ambiguous $\sigma_u$ when adding the final tile $a_\zeta$. 

\subsection{Production of Specific Shapes with the Backboned aTAM}
\label{section: individual_kits}

Having derived universal assembly kits for different variants of the backboned aTAM, we now consider how different assembly models perform when assembling specific shapes. In doing so, we hope to further verify our intuition that sequence information facilitates the assembly of shapes with fewer distinct tile types, and also observe differences in the performances of different sequence-directed assembly models. We consider the following three assembly models in this section. 

\begin{enumerate}
\item Backboned aTAM with arbitrary interactions ({\it i.e., allowing neutral inter-tile interactions}).
\item Backboned aTAM with only attractive inter-tile interactions.
\item Sequence-free polyomino assembly (temperature-1 aTAM with tile rotation).
\end{enumerate}

\begin{figure}[t]
\centering
\begin{subfigure}{0.4\textwidth}
\centering
  \caption{\label{fig:rectangle}}
  \includegraphics[width=\linewidth]{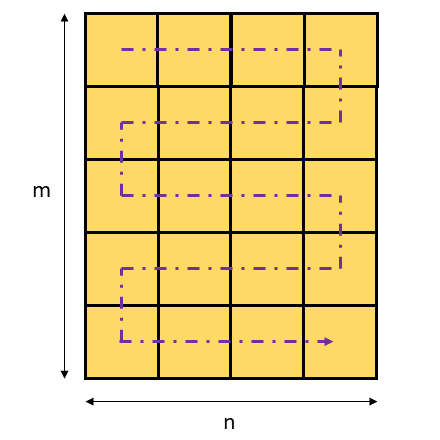}
\end{subfigure}
\begin{subfigure}{0.4\textwidth}
\centering
  \caption{\label{fig:bulge_rectangle}}
  \includegraphics[width=\linewidth]{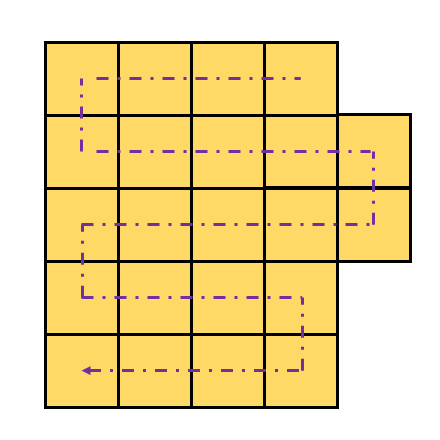}
\end{subfigure}
\caption{ We consider the construction of rectangles \textbf{(a)} and bulged rectangles \textbf{(b)} with $m$ rows and $n$ columns. The backbone path for the backboned aTAM is given in purple. The bulge location is fixed to the second and third rows of the rightmost column. }
\label{fig:individual_assembly_kit}
\end{figure}

\begin{figure}[t]
\centering
\begin{subfigure}{\textwidth}
\centering
  \caption{\label{fig:rectangle_backbone}}
  \includegraphics[width=\linewidth]{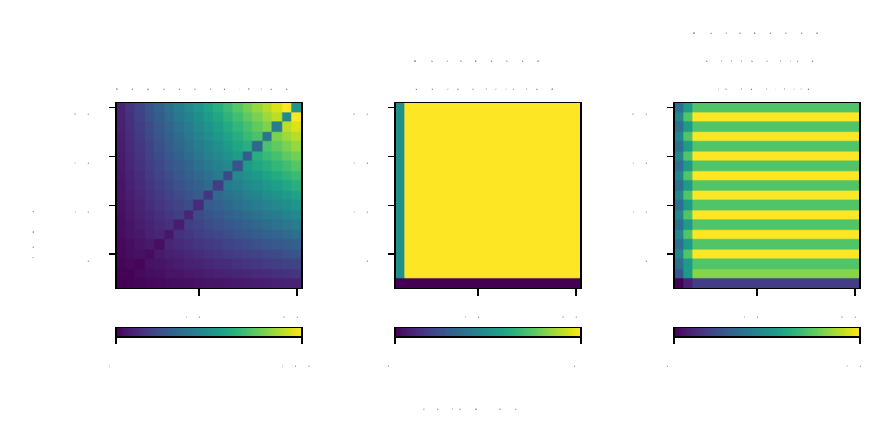}
\end{subfigure}
\begin{subfigure}{\textwidth}
\centering
  \caption{\label{fig:bulge_rectangle_backbone}}
  \includegraphics[width=\linewidth]{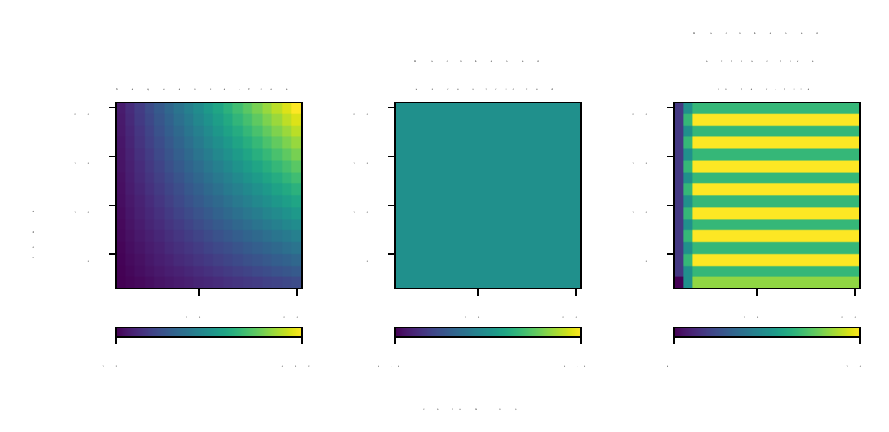}
\end{subfigure}
\caption{ Upper bounds on the tile complexity, or the number of tiles needed to construct \textbf{(a)} rectangles and \textbf{(b)} bulged rectangles of various row and column numbers using different assembly maps.  Backboned maps tend to be capped in the number of tiles they require for arbitrarily large shapes (of the given types), while the sequence-free assembly map requires an unbounded number of tile types. Forcing attractive inter-tile interactions tends to increase the number of tile types required. }
\label{fig:individual_assembly_kit}
\end{figure}

For these models, we consider upper bounds on the minimum number of tiles, also called the tile complexity, needed to construct two types of target shape, rectangles and rectangles with a bulge (Figures \ref{fig:rectangle} and \ref{fig:bulge_rectangle}), starting from the empty configuration. We apply the method in \cite{Ahnert2010} to upper bound a minimal tile complexity for a given shape using sequence-free polyomino assembly. For the backboned aTAM with arbitrary interactions, an upper bound on the minimum tile set is easy to derive, as we can follow some arbitrary hamiltonian path and use the tile types in Theorem \ref{theorem:walk_main}. 
%It is clear that, for this restricted case, the sequenced aTAM will assemble the same shape {\color{red}using the same assembly kit}. 
For the backboned aTAM with only attractive inter-tile interactions, we can follow a face assignment algorithm implied by Theorem \ref{theorem:interacting_walk} to find an upper bound on the tile complexity. For these assembly schemes, the hamiltonian path taken is assumed to zig-zag down the rectangle along each row starting from the top row of the rectangle (Figure \ref{fig:individual_assembly_kit}\,(a) and \ref{fig:individual_assembly_kit}\,(b)). 

 In Fig.~\ref{fig:individual_assembly_kit}, we plot upper bounds on the minimum tile complexities of each target shape using different assembly maps. The very smallest rectangles can be constructed with fewer unique tile types using sequence-free assembly than with backbone-directed assembly, as sequence-free assembly is better able to utilize symmetries within the shape \cite{Greenbury2014}. As expected, the asymmetric bulged rectangles display no such behavior as they have no symmetries of which the base aTAM can take advantage. In both cases, the number of tiles required for the base aTAM increases unboundedly for large rectangles, as anticipated. By contrast, the backboned aTAM models show a growth in the number of tile types that plateaus, verifying our intuition that sequences facilitate the assembly of large shapes with a small number of building blocks. 
 %Comparing sequence-directed maps, we see that maps with neutral interactions perform best, as they have the fewest assembly constraints. Interestingly, our upper bound for sequenced aTAM is lower than for backboned aTAM {\color{red} for this class of shape} when both are constrained to have attractive inter-tile interactions.

\subsection{The Sequenced aTAM has No Universal Assembly Kit}

 We have shown that the backboned aTAM possesses a universal assembly kit. We have also confirmed by example that, for large structures, one can achieve deterministic assembly of a single structure with smaller tile sets. At least in part, this difference arises because each instance of the backboned aTAM is associated with a sequence that acts as an information-containing program, in addition to the rules of interaction quantified by the glues. However, an additional difference is provided by the geometric constraints on growth provided by the backbone. 

 In an attempt to disentangle the role of sequence information and geometric constraints, we have introduced the sequenced aTAM defined in Section \ref{sec:model_def}.
We will now show that no universal assembly kit exists that allows the sequenced aTAM to deterministically assemble every shape in its shape space, regardless of constraints on neutral versus attractive inter-tile interactions. 

%is preserved under more stringent conditions, suggesting that the presence of a universal kit could be due to a deeper property of the model rather than the simplicity induced by our original construction. It is natural to think that this is exclusively due to the information content of a sequence, since the set of self-avoiding walks of a certain size and the set of sequences both grow exponentially. While this in no doubt contributes to the ability of the backboned aTAM to exhibit a universal assembly kit, we do not believe that this is the full story. In particular, we show that by weakening a key assumption of the backboned aTAM, leading to the sequenced aTAM defined in \ref{sec:model_def}, we lose this universal assembly kit. As a reminder, sequenced aTAM behaves similarly to backboned aTAM, except added tiles are no longer constrained to neighbour immediately preceding tiles. 

\begin{figure}
    \centering
    \includegraphics[width=0.8\textwidth]{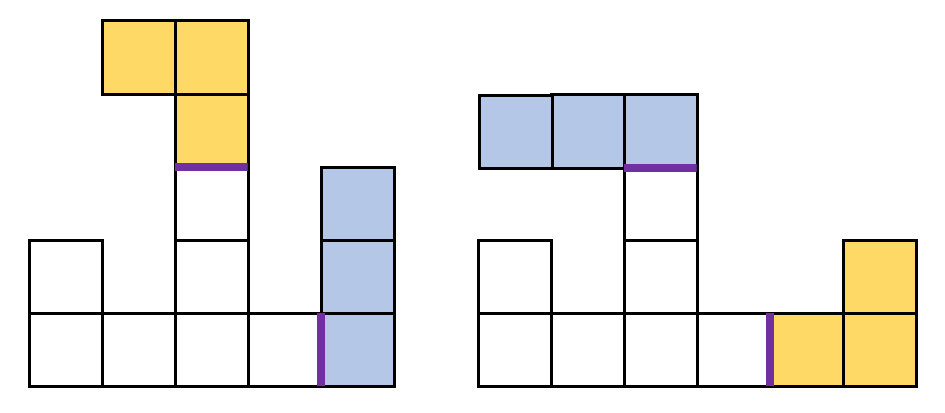}
    \caption{ Intuition for Lemma \ref{theorem:the_beginning}. \textbf{(Left)} A configuration $A$ given some starting configuration $A_0$ (white), and subconfigurations $A_\epsilon$ (blue) and $A_{\epsilon'}$ (yellow). \textbf{(Right)} Assume the faces in purple have the same glue, then subshapes that grow from the purple faces can be flipped while maintaining the same tile type sequence $Q$, forming an alternate configuration $A'$.}
    \label{fig:exchange}
\end{figure}

Our proof proceeds from the following intuition. For an instance of the sequenced aTAM, shapes that grow from a face of a configuration $A_t$ for times $t^\prime>t$ are encoded fully by the glue of the face and the tile type sequence past time $t$ (if we ignore the `blocking' of growth by preexisting tiles). Hence, if we have two tile faces with the same glue type at a time $t$, then (ignoring tile blocking), we cannot stop a shape that grows from one of the faces from also growing on the other face (Figure \ref{fig:exchange}). Hence, if we need to grow $N$ different shapes from $N$ different faces at a time $t$, we require at least $N$ different glue types. If the growth of certain shapes requires an arbitrarily large number of open faces, $N\rightarrow \infty$, then a finite universal assembly kit will be impossible. 

This intuition is incomplete for the following reasons:
\begin{enumerate}
    \item We haven't considered the effect of tile blocking, which can potentially allow many different shapes to grow from the same glue, based on the blocking pattern.
    \item We haven't shown that, for certain classes of shapes, a large number of open faces are unavoidable. A requirement for a large number of open faces is not trivial. For example, when constructing rectangles of any size, it is possible to avoid opening more than two non-neutral faces at any one time by constructing the rectangle row-wise. 
\end{enumerate}

Lemma \ref{theorem:the_beginning} formalizes this intuition. Using this lemma, Theorem \ref{theorem:infinite_interactions} in the main text builds a set of assumptions that allow us to avoid tile blocking and forces us to `open' at least $N$ distinct glue types when $g \geq 0$. This result is extended to unrestricted $g$ in Theorem \ref{theorem:InfiniteInteractions} in the appendix. We then develop a class of shapes that obey these assumptions for arbitrarily large $N$ in Definition \ref{def:treeangle}, and complete the proof in Theorem \ref{theorem:existence}. Readers interested only in the consequences of this result in terms of the efficiency with which the variant aTAM models exploit information within their sequence programs may skip directly to Section \ref{sec:Kolmogorov}. 

We first state a few additional definitions to aid in our proof. Let the \textbf{face coordinate} $ (z,k)$ for a 2D coordinate $z$ and orientation $k \in \{N,E,S,W\}$ denote the $k$-facing face of the coordinate $z$. Then, let $A(z,k)$ return the glue type of $(z,k)$. Two configurations $A$ and $A'$ are \textbf{adjacent} if there exist $z \in dom(A)$ and $z' \in dom(A')$ such that $z$ neighbours $z'$. It is always possible to define a \textbf{face set} of a configuration $face(A)$ as the set of faces of all coordinates in $dom(A)$. Two neighbouring faces $(z,k)$ and $(z',k')$ form a \textbf{face-pair}.  Two configurations $A$ and $A^\prime$ are said to be \textbf{uniquely adjacent} through a face-pair $((z,k),(z',k'))$ if it is the only face-pair where one face of the pair is in the face set of each configuration. 

Over the course of our proofs, we will be relying on sufficient and necessary conditions for some configuration $A_t$ to be part of a complete trajectory. We note that, from the definitions of the backboned and sequenced aTAM, it is clear that $g-validity$ is a necessary condition for any configuration in a complete trajectory. Together with the consistency of starting configurations and added tiles matching the sequence $Q$ (either to the end of $Q$ or until the point of premature termination), g-validity becomes sufficient in ensuring that a given trajectory is complete. We now proceed by formalizing our intuition that open faces growing distinct subshapes require distinct faces, under some set of assumptions:

\begin{lemma}
Consider a sequenced aTAM instance $\mathcal{P} = (A_0,Q,g)$. Let the following be true:
\begin{enumerate}
\item Let $A_0(z_\epsilon,k_\epsilon) = A_0(z_{\epsilon'},k_{\epsilon'})$ for two faces $(z_\epsilon,k_\epsilon)$ and $(z_{\epsilon'},k_{\epsilon'})$.
\item Let $A_{\epsilon}$ be a configuration adjacent to $A_{0}$ uniquely through face pair $(\dot{e}_{\epsilon},\ddot{e}_{\epsilon})$ ($\dot{e}_{\epsilon}$ in $A_\epsilon$ and $\ddot{e}_{\epsilon}$ in $A_0$), while $A_{\epsilon'}$ is a configuration adjacent to $A_{0}$ uniquely through face pair $(\dot{e}_{\epsilon'},\ddot{e}_{\epsilon'})$ and assume $A_{\epsilon}$ and $A_{\epsilon'}$ are not adjacent to each other. Assume further that some complete trajectory $\Psi = (A_0, A_1, A_2, ..., A_{\zeta}) $ generated by $T$ exists such that $A_{\zeta}  = A_0 + A_{\epsilon} + A_{\epsilon'} + A_c $ for some arbitrary configuration $A_c$ not adjacent to $A_{\epsilon}$ or $A_{\epsilon'}$.  
\item Let $A_{\epsilon \rightarrow \epsilon'}$ be an affine transformation of $A_\epsilon$ that maps $\dot{e}_\epsilon$ to $\dot{e}_{\epsilon'}$, and similarly for $A_{\epsilon' \rightarrow \epsilon}$. Then, assume that $A_{\zeta}' = A_0 + A_{\epsilon \rightarrow \epsilon'} +
A_{\epsilon' \rightarrow \epsilon} + A_{c}$ is $g$-valid.
\end{enumerate}
Then, there exists another trajectory $\Psi' = ( A_0', A_1', ..., A_{\zeta}',... )$ complete with respect to $\mathcal{P}$ such that the terminal configuration of $\Psi'$ is $A_{{\zeta}}'$ or has $A_{{\zeta}}'$ as a subconfiguration.
\label{theorem:the_beginning}
\end{lemma}

\begin{proof} We show that some selection of $A_t'$ for $\Psi'$ meets the definition of a complete trajectory. First, by setting $A_0 = A_0'$, starting configuration consistency is established. Configurations in the trajectory $\Psi$ can be decomposed as $A_t = A_0 + A_{\epsilon,t} + A_{\epsilon',t} + A_{c,t}$, where $A_{x,t}$ is a subconfiguration of $A_{x}$ (present at time t). We require that it is possible to construct a trajectory $\Psi'$ complete with respect to $T$, such that each entry $A_{t}' = A_0 +  A_{\epsilon \rightarrow \epsilon' ,t} +
A_{\epsilon' \rightarrow \epsilon,t} + A_{c,t}$ up to $t = \zeta$ is g-valid. This condition is always true at $t = 0$, and tile addition at time $t+1$ (following the tile type sequence $Q$) can always result in a g-valid
$A_{t+1}' = A_0 +  A_{\epsilon \rightarrow \epsilon' ,t+1} +
A_{\epsilon' \rightarrow \epsilon,t+1} + A_{c,t+1}$ up to $t+1 = \zeta$. Since the order of tiles added also matches the sequence $Q$, the condition for complete trajectories of sequenced aTAM is obeyed, and the theorem is thus true by induction.
\end{proof}

In essence, we have made the general argument that if the left hand configuration in Fig.~\ref{fig:exchange} is complete with respect to $T$, and swapping over the yellow and blue sub-configurations does not result in a clash, then the right hand configuration can also be formed under the dynamics of $T$.

Some additional definitions will be useful at this point as we leverage Lemma~\ref{theorem:the_beginning} to show that deterministic growth of certain shapes with $M$ protrusions from an initial subshape requires at least $\frac{M}{4}$ unique tiles. Let an \textbf{open face} be a face of some non-empty tile in a configuration that neighbours an empty face, and glue types belonging to some open face are similarly known as open glue types. Let $\mathbf{E(A)}$ return the oriented shape of a configuration.

\begin{theorem}
\label{theorem:infinite_interactions}
Consider an oriented shape $S = \bigcup_{i = 0}^M S_i$ such that: 
\begin{enumerate}
     \item For any $i,j \in {1,..,M}$, $S_i \not\simeq S_j$ for any $i \neq j$. That is, each $S_i$ is distinct including rotations and translations. \label{step:no_equivalence}
    \item For any $i,j \in {1,..,M}$, $S_i$ has no point that neighbours any point in any other $S_j$ where $ i \neq j$. \label{step:no_neighbour}
    \item For any $i \in {1,..,M}$, $S_i$ has exactly one coordinate that neighbours a coordinate in $S_0$, and $S_0$ has exactly one coordinate that neighbours this coordinate. We denote by $\dot{z}_i$ the coordinate in $S_i$ and by $\ddot{z}_i$ the neighbouring coordinate in $S_0$ (neighbouring faces are similarly labeled ($\dot{z}_i$, $\dot{k}_i$) and ($\ddot{z}_i$, $\ddot{k}_i$)). 
    \label{step:inter_connection}
    \item $S$ and $S_0$ are not rotationally symmetric. Additionally, there does not exist $S_t$, a subshape of $S$, such that $S_t \neq S_0$ but $S_t \simeq S_0$. \label{step:oriented_det}
\end{enumerate}
Consider a sequenced aTAM instance $\mathcal{P} = (A_0, Q, g)$ such that $E(A_0) = S_0$ and \textbf{the range of $\mathbf{g}$ is restricted to $\mathbf{g \geq 0}$}. If assembly is deterministic and the terminal configuration $A_{\zeta}$ is such that $E(A_{\zeta}) = S$ as defined above, then the number of unique tile types required in $A_0$ grows with  at least $\frac{M}{4}$.
\end{theorem}

\begin{proof} 
Condition \ref{step:oriented_det} and determinism imply that $\mathcal{P}$ is oriented-deterministic. Assume $A_0$ has fewer than M distinct open glue types, so some $S_i$ and $S_j$ must be anchored at glues of the same type. Then, Theorem \ref{theorem:the_beginning} implies one of the following must be true:
\begin{enumerate}
\item $S_i \simeq S_j$.
\item An overlap of configurations occurs when transforming the configurations neighbouring the two faces $(\ddot{z}_i, \ddot{k}_i)$ and $(\ddot{z}_j, \ddot{k}_j)$ (in the language of Theorem \ref{theorem:the_beginning}, $A_{\zeta}' = \infty$ and is thus not $g$-valid).
\end{enumerate}

Each of these possibilities leads to a contradiction. The first possibility directly contradicts condition \ref{step:no_equivalence}. The second possibility results in some point of $S_i$ (or $S_j$) neighbouring some additional point in $S$ outside of $S_i$ (or $S_j$), contradicting condition \ref{step:no_neighbour} or \ref{step:inter_connection}. Hence, by contradiction, $A_0$ must have $M$ or more open faces. As each tile type has at most $4$ unique faces, the number of unique tile types in $A_0$ must grow with at least $\frac{M}{4}$. 
\end{proof}

The bolded restriction on $g$ in Theorem \ref{theorem:infinite_interactions} excludes the possibility of repulsive interactions. Thus we do not have to consider situations where a tile is blocked by some repulsive interaction, rather than overlapping. We have chosen to include this simplified form of the theorem in this text as it provides a better intuition for our proof. However, with a few additional assumptions, we show how this restriction can be lifted in Theorem \ref{theorem:InfiniteInteractions}, included in Appendix A, and the remainder of our results are consistent with this more general form.

We now construct a class of target shapes such that shapes obeying the assumptions in Theorem \ref{theorem:infinite_interactions} (and more strongly, those in Theorem \ref{theorem:InfiniteInteractions}) cannot be avoided in the assembly of these targets. To do so, we define a \textbf{branching point} as any coordinate of some oriented shape $S$ with three or more neighbours in $S$, while a \textbf{corner} is a coordinate in $S$ with two neighbours in $S$, such that the corner and its two neighbours do not form a straight line. Then, let a 
\textbf{straight line segment} be a line starting at some branching point or corner, and ending at the next branching point/corner. The \textbf{distance} between two points $(x_1,y_1)$ and $(x_2,y_2)$ is taken to be $|x_2-x_1|+|y_2-y_1|$.
We can now begin describing the fundamental components of our constructed target shape.

\begin{figure}
    \centering
    \includegraphics[width=\textwidth]{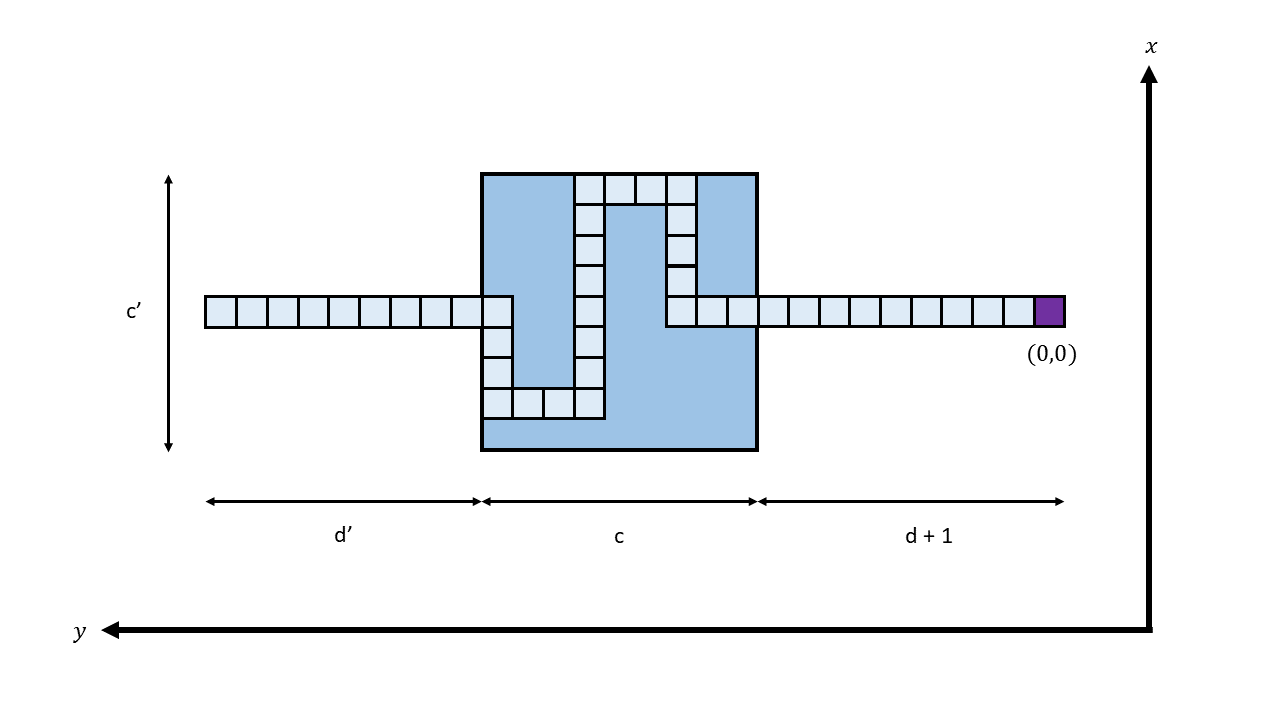}
    \caption{An oriented shape $W$ belonging to the set $\mathbb{W}_{9,9,9,9}$. The \textbf{start point} of the oriented shape is given in purple. The dark blue shaded box is the \textbf{canvas} of $W$, the region of $W$ containing an arbitrary self-avoiding walk.}
    \label{fig:straight_tipped_walk}
\end{figure}

\begin{definition}
\label{def:straightwalk}
An oriented shape $W$ is in the set $\mathbb{W}_{d,c,c',d'}$ for arbitrary positive integers $d$,$c$ and $d'$ and odd positive integer $c'$ if:
\begin{enumerate}
\item It is a 2-Dimensional self-avoiding walk constrained to a length of $c'$ in one dimension and a length of $d + c + d'$ in the other dimension. Irrespective of the true orientation of the walk in the 2D plane, $c'$ is called the width of the oriented shape (with corresponding dimension called the width dimension) and $d + c + d'$ is called the height of the shape (with corresponding dimension called the height dimension).
\item The first $d+1$ and last $d'$ coordinates (along the height dimension) form straight lines that, if extended, would cut through centre of the central $c \times c'$ rectangle, called the canvas.
\item The canvas contains some self-avoiding walk that doesn't leave the canvas and connects the two straight lines.
\item Any point in $W$ is of at least distance 3 from any other point in $W$, except for the 2 points immediately preceding it and the 2 points immediately succeeding it along the walk. 
\item Any straight line segments in the walk $W$ must be of at least length 4. 

\end{enumerate}
Refer to figure \ref{fig:straight_tipped_walk} for an example of an oriented shape in $\mathbb{W}_{d,c,c',d'}$.
\end{definition}

Note that features 3 and 4 are not necessary for Theorem \ref{theorem:infinite_interactions}, but are necessary for the target shape we are constructing to obey the assumptions of Theorem \ref{theorem:InfiniteInteractions}. We use these oriented shapes to build larger oriented shapes with the aim of creating a class of shapes that fulfill the assumptions laid out in Theorem \ref{theorem:infinite_interactions} (as well as Theorem \ref{theorem:InfiniteInteractions}). We call these larger oriented shapes treeangles, and we define them as follows.

\begin{figure}
    \centering
    \includegraphics[width=\textwidth]{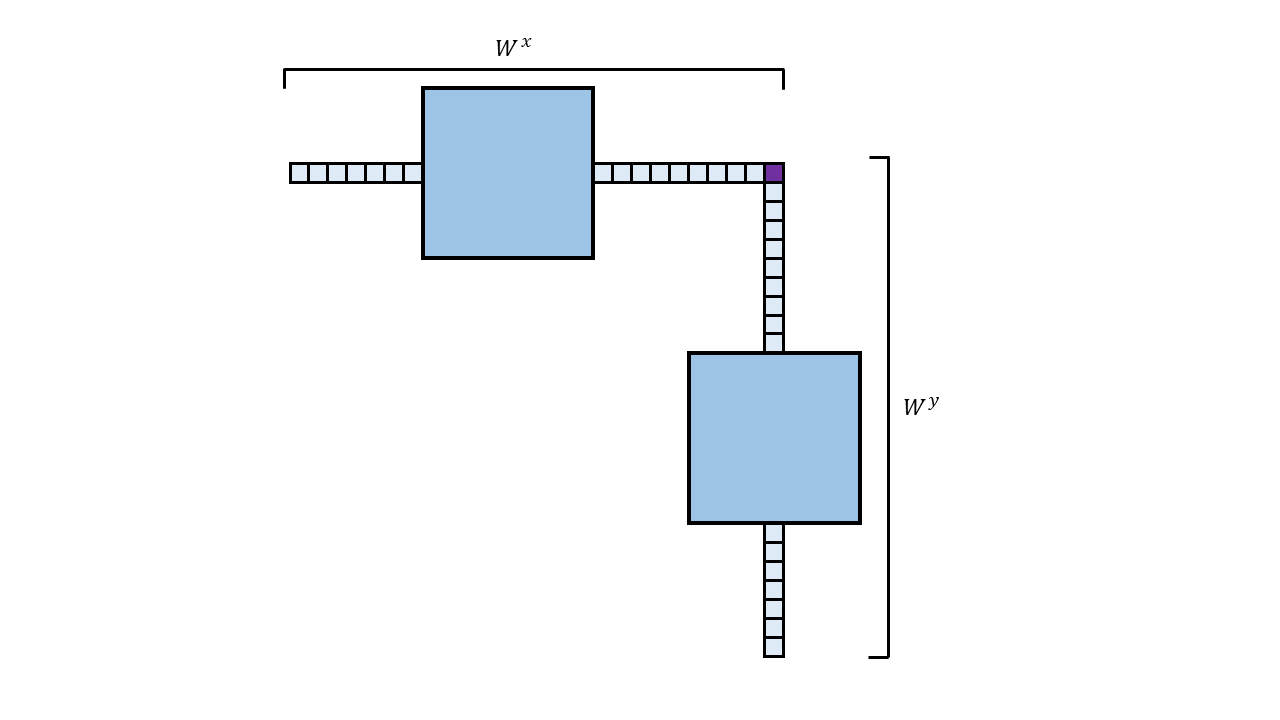}
    \caption{An oriented shape $R$ belonging to the set $\mathcal{R}^0_9$, consisting of $W^x$ a horizontally aligned member of $\mathbb{W}_{9,9,9,7}$ joined with $W^y$ a vertically aligned member of $\mathbb{W}_{9,9,9,7}$. Shaded blue boxes represent canvas regions of constituent walks. Each of these boxes are constrained to contain distinct self-avoiding walks. }
    \label{fig:R0L}
\end{figure}

\begin{definition}
A $treeangle$ of order $V$ with fundamental length $L$ is an oriented shape $S$ that can be constructed as follows.
\begin{enumerate}
    \item An oriented shape in the set $\mathcal{R}^0_L$ can be constructed by placing a horizontal oriented shape (height dimension along $x$) $W^x \in \mathbb{W}_{L,L,L,L-2}$ and a vertical oriented shape (height dimension along $y$) $W^y \in \mathbb{W}_{L,L,L,L-2}$ next to each other, with the two connected at their start points (Figure \ref{fig:R0L}).
    \item Given the Definition for $\mathcal{R}^0_L$, $\mathcal{R}^1_L$ can be constructed by connecting a horizontally aligned $W^x \in \mathbb{W}_{L,L,L,L}$ with a vertically aligned $W^y \in \mathbb{W}_{L,L,L,L}$, and then connecting two newly sampled oriented shapes from $\mathcal{R}^0_L$ at the endpoints of $W^x$ and $W^y$  (Figure \ref{fig:R12L}).
    \item For higher orders $v$, $\mathcal{R}^v_L$ is obtained in a similar way from 
    $\mathcal{R}^{v-1}_L$. However, to ensure that the oriented shapes fit in the 2D lattice, $W^y \in \mathbb{W}_{2^{v-1}L,2^{v-1}L+2^{v-1}-1,L,2^{v-1}L}$  (Figure \ref{fig:R12L}).
    \item To construct a treeangle $S$, sample two oriented shapes $R_V$ and $\hat{R}_V$ from $\mathcal{R}^V_L$. $\hat{R}_V$ is horizontally flipped and translated, forming $\tilde{R}_V$ and $I \in \mathbb{W}_{2^{v-1}L,2^{v-1}L+2^{v-1}-1,L,2^{v-1}L}$ is sampled to connect $R_V$ and $\tilde{R}_V$ (Figure \ref{fig:treeangle}).   
    \item During construction, every walk sampled from some $\mathbb{W}_{d,c,c',d'}$ is sampled without replacement. That is, each section of the treeangle (defined as any contiguous path between two branching points) has a unique walk. 
    \item A shape $\overline{S}$ consisting of rotations and translations of some (oriented) treeangle $S$ is also known as a treeangle. The order and fundamental length of $\overline{S}$ are equal to the order and fundamental length of $S$.
\end{enumerate}
\label{def:treeangle}
\end{definition}

\begin{figure}
    \centering
    \includegraphics[width=\textwidth]{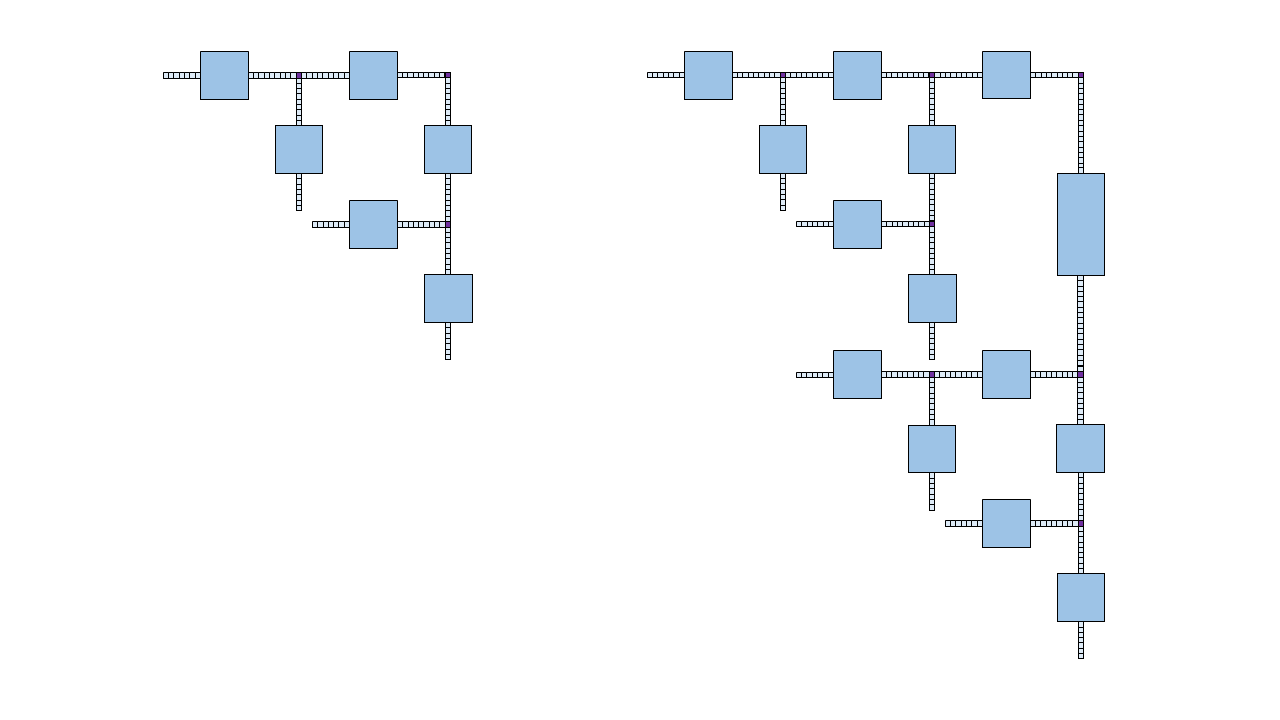}
    \caption{ Oriented shapes $R_1$ (left) and $R_2$ (right) belonging to $\mathcal{R}^1_9$ and $\mathcal{R}^2_9$ respectively. Oriented shapes in $\mathcal{R}^2_9$ can be obtained starting from two oriented shapes in $\mathcal{R}^1_9$ by joining them via a horizontal walk $W^x \in \mathbb{W}_{9,9,9,9}$ and a vertical walk $W^y \in \mathbb{W}_{18,19,9,18}$. Shaded blue boxes represent canvas regions containing distinct self-avoiding walks. }
    \label{fig:R12L}
\end{figure}

\begin{figure}
    \centering
    \includegraphics[width=\textwidth]{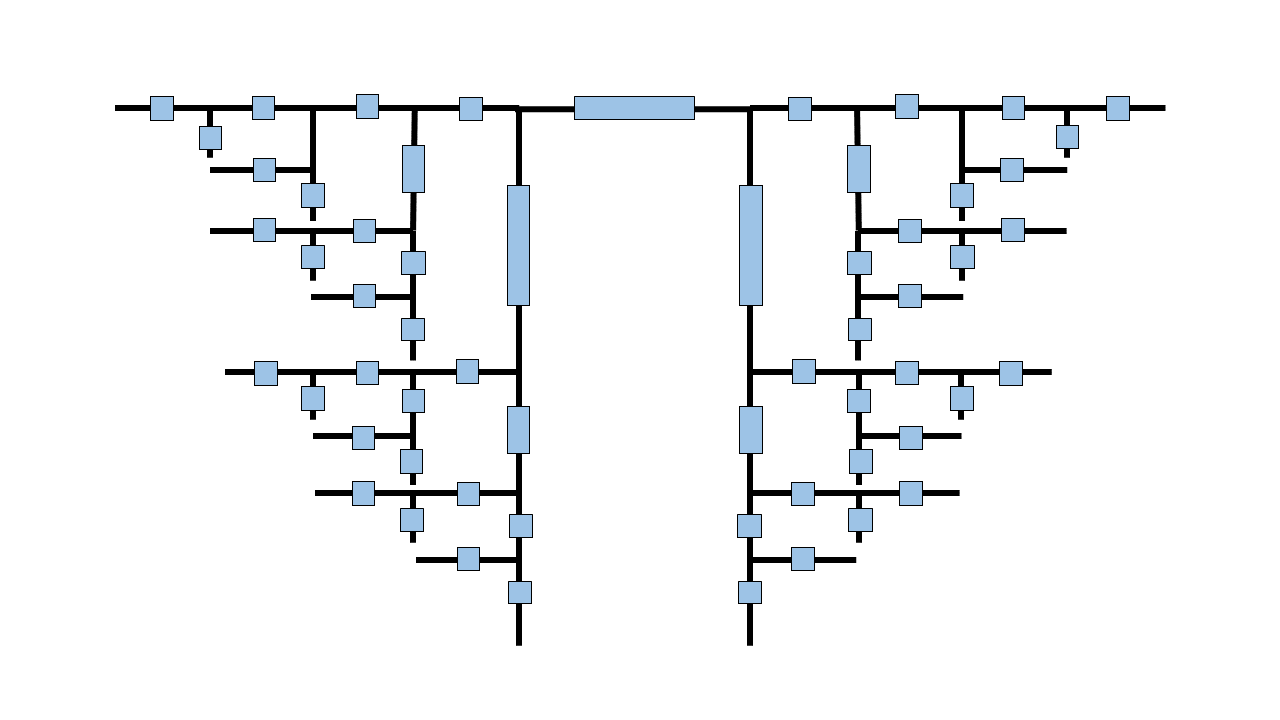}
    \caption{ An abstract schematic of a treeangle of order 2. Black lines represent straight segments of tiles, while blue rectangles represent canvases, each containing distinct self-avoiding walks. }
    \label{fig:treeangle}
\end{figure}

Using the treeangle, we can now set out to prove the inexistence of a universal assembly kit for the sequenced aTAM.  Before doing so, we lay out a few final definitions. Let $\mathbf{G_S(S)}$ return a graph whose vertices represent coordinates of $S$, and where edges are drawn between neighbouring coordinates. Let a \textbf{terminal branching point} be a branching point $v$ such that a path $P$ on the graph $G_S(S)$ can be drawn from some leaf node (a node with only one edge in a graph) to $v$ such that $P$ contains no other branching point (that is, it a branching point where at least one of the paths from it to an end of the branch includes no further branching points). Consider an oriented shape $S' \subset S$. An \textbf{external neighbour} of $S'$ is a coordinate in $S$ but not in $S'$ that neighbours a coordinate in $S'$. We may now proceed to our final set of proofs in which we finally prove that the sequenced aTAM requires an unbounded number of tile types to assemble arbitrary treeangles, and thus cannot admit a universal assembly kit.

%%An \textbf{incomplete branching point} is a branching point with one of its branches removed -  we say it is a branching point in the original oriented shape $S$, but an incomplete branching point in the oriented shape $S'$ with the branch removed. 

\begin{lemma}
Consider a sequenced aTAM instance $\mathcal{P} = (A_{empty}, Q, g)$ that deterministically assembles a treeangle $\overline{S}$ of order $V$. Then, the number of unique tile types required in $Q$ grows with at least $\frac{V}{4}$. 
\label{theorem:scale}
\end{lemma}

\begin{proof} For simplicity, consider a specific oriented shape $S \in \overline{S}$ and consider the graph $G_S(S)$. Let $R_V$, $I$ and $\tilde{R}_V$ be defined as in Definition \ref{def:treeangle}. Then, the subshape defined by any path $P$ from $R_V$ or $I$ to a terminal branching point of $\tilde{R}_V$, excluding the terminal branching point, must have at least $V$ external neighbours. Consider then a shape $S'$ with $P$ as a subshape, such that $S'$ does not contain any terminal branching points of $R_V$. Without loss of generality, starting assembly at any coordinate of $I$ or $R_V'$, any complete trajectory $\Psi$ of $\mathcal{P}$ that assembles $S$ must admit an element $A'$ such that $E(A') = S'$ for some $S'$ (an analogous argument applies for trajectories starting at a coordinate of $R_V$), since there must exist some time point where the first terminal branching point is incorporated onto $S$. Furthermore, $S'$ can have no fewer than $V$ external neighbours, because the inclusion of further branching points of $R_V$ can only ever increase the number of external neighbours of $S'$, and since $S'$ does not contain a terminal branching point of $R_V$, then the number of external neighbours cannot be decreased by the addition of more coordinates onto $S'$. By the construction rules of the treeangle, a partition of $S$ as $S = \bigcup_{i = 0}^M S_i$ following Theorem \ref{theorem:infinite_interactions} (as well as Theorem \ref{theorem:InfiniteInteractions} in the Appendix) is then possible with $S_0 = S'$ and $M \geq V$, and the remaining rules are guaranteed by our treeangle construction. Letting $Q'$ be the subsequence of $Q$ after the assembly time in which $A'$ appears, the sequenced aTAM instance $(A',Q',g)$ requires at least $\frac{V}{4}$ unique tile types in $A'$ to assemble $\overline{S}$ deterministically. Since all non-empty tile types in $A'$ must appear in $Q$, $Q$ requires at least $\frac{V}{4}$ unique tile types, completing the proof. 
\end{proof}

%%{\color{red}[Theorem 6 requires each $S_i$ to contain a branch point, but it's not clear to me that you've defined your incomplete branching points so that they are non-terminal themselves. So do the paths contain at least $V$ non-terminal incomplete branching points?]}

%%{\color{red}[Any trajectory must contain a state with $V$ incomplete branching points. It seems to me that you can always construct a trajectory so that it only contains at most 1 incomplete branching point, by always adding an extra tile immediately after you go through a branching point. Clearly there's still an incomplete stub of a branch here, and the logic still applies (I think) but I suspect you need to slightly revise what you mean by incomplete branching point? A non-terminal branching point in the target shape which looks like a terminal branching point (or not a branching point at all) in the current shape?]}

\begin{theorem}
The sequenced aTAM does not admit a universal assembly kit.  
\label{theorem:existence}
\end{theorem}
\begin{proof} We only need to show that given any positive integer $V$, some treeangle $\overline{S}$ of order $V$ exists. The number of walks in any $\mathbb{W}_{L,L,L,L}$ increases monotonically with $L$. Each $\mathbb{W}_{L,L,L,L}$ is non-empty as a simple straight line fulfilling length constraints will be in this set. Increasing $L$ monotonically increases the number of shapes in $\mathbb{W}_{L,L,L,L}$, as we can simply extend the walk in the canvas with another straight segment or some number of curved walks. Hence, there is always some $L$ that will provide sufficient walks to generate a treeangle $\overline{S}$ of arbitrary order $V$. This argument extends to other sets $\mathbb{W}_{d,c,c',d'}$. The proof then follows from Lemma \ref{theorem:scale}. 
\end{proof}

\subsection{Geometric Constraints, Shape Space Size and Kolmogorov Complexity}
\label{sec:Kolmogorov}

 We now frame our results in the context of the shape space and worst-case Kolmogorov complexity scaling of our considered assembly maps. Sequences of length $N$, drawing from an alphabet of size $M$, have $M^N$ unique variants, and hence can in principle deterministically encode up to $M^N$ shapes. Similarly, an arbitrary sequence of length $N$ requires $N \log{M}$ bits, and hence the ``program'' defined by a backboned or sequenced aTAM exceeds $N \log{M}$ in length, scaling as $O(N)$ if $M$ is fixed.

 The backboned aTAM's shape space is equal to the set of self-avoiding walks, which grows with shape size $N$ as $O(2.638^N)$ from \cite{Clisby2012}. Further, as self-avoiding walks can be encoded by left-right-forward moves, the scaling in worst-case Kolmogorov complexity is upper bounded by $O(N)$. Consistent with Theorems \ref{theorem:walk_main} and \ref{theorem:interacting_walk}, these scaling results do not preclude a universal assembly kit of size $M \geq 3$. Indeed, the backboned aTAM without restrictions gets very close to a universal assembly kit with $M=4$, perhaps unsurprisingly as it simply routes the underlying hamiltonian path.  Constraining the assembly mechanism so that all interactions must be attractive does not prevent the existence of a universal assembly kit, but our upper bound on the kit size --  over $200$ tiles -- suggests that the information in the sequence is used much less efficiently. Although the true minimal universal assembly kit size may be smaller than 200 tiles, it will likely share the same two sources of inefficiency: firstly, many sequences will form equivalent structures; and secondly, others will not deterministically form a well-defined structure. 
 
 It is interesting to note the interplay between these two sources of inefficiency. Imposing more demanding rules on the assembly mechanism forces more tiles to be used to specify the same set of backbone routings, in order to avoid non-deterministic or incomplete assemblies. In turn, this larger set of tiles leads to a redundancy, whereby many sequences deterministically assemble the same structure. 

 The shape space of the sequenced aTAM is equal to the set of free polyominoes (polyominoes equivalent under rotation), with a scaling upper bound given by $O(4.5252^N)$  \cite{Barequet2022}  (this upper bound is presented for fixed polyominoes, but note the numbers of fixed and free polyominoes are related by a constant factor in the large $N$ limit \cite{Redelmeier1981}). We now also argue that the scaling of the worst case Kolmogorov complexity of shapes in the shape space of sequenced aTAM is upper bounded by $O(N)$. 
 
 While to our knowledge this argument for the $O(N)$ scaling has not been explicitly made in literature, it follows trivially from Klarner's algorithm for polyomino enumeration \cite{Klarner1967}. A brief review of Klarner's algorithm is provided as Algorithm \ref{alg:Klarner}. Starting at the lowest tile of the leftmost column of a fixed polyomino $S$, Klarner's algorithm iteratively connects new tiles to neighbouring ``parent'' tiles, creating a queue of tiles as it goes. The tiles in the queue are successively tested for new neighbours. Moreover, each tile is labelled by going clockwise around the tile, starting from (but excluding) the given tile's parent.  The label is a 3-bit description, storing a value `$1$' if a tile's neighbor is present in $S$, and `$0$' otherwise, for each neighbour (Figure \ref{fig:Klarner}). 

\begin{algorithm}
\caption{Klarner's algorithm for uniquely labelling polyominoes. A $3N$-bit description of a given input polyomino of size $N$ is produced as output.} \label{alg:Klarner}
\SetKwFunction{GetNeigh}{getNeighbors}
\SetKwProg{Fn}{Function}{:}{}

\Fn{\GetNeigh{$S,\mathrm{current},\mathrm{parent},\&\mathrm{tileQueue},\&\mathrm{parentMap}$}}{
    triple $\leftarrow (0,0,0)$\;
    \For{$i \leftarrow 1$ \KwTo $3$}{
        neighbor $ \leftarrow$ neighbor $i$ of current, counting clockwise with parent at $i = 0$\;
        \If{$\mathrm{neighbor} \in S \mathrm{\;and\;} \mathrm{parentMap}(\mathrm{neighbor}) = \mathrm{empty}$}{
            triple$[i] \leftarrow 1$\;
            parentMap$($neighbor$) \leftarrow $current\;
            enqueue$($tileQueue,neighbor$)$
        }
    }
    \Return{$\mathrm{triple}$}\;
}

\KwIn{Fixed polyomino $S$}
$N \leftarrow$ size$(S)$\;
current $\leftarrow$ lowest tile position of the leftmost column of $S$\;
parent $\leftarrow$ position below current\;
tileQueue $\leftarrow \emptyset$\;
parentMap $\leftarrow \mathrm{Map}\{\mathrm{current'} \to \mathrm{parent'}\}$\;
result $\leftarrow $ array of $N$ triples $(0,0,0)$\;

\For{$i \leftarrow 1$ \KwTo $N$}{
    result$[i] \leftarrow$ \GetNeigh{$S,\mathrm{current},\mathrm{parent},\mathrm{tileQueue},\mathrm{parentMap}$}\;
    current $\leftarrow$ dequeue$($tileQueue$)$\;
    parent $\leftarrow$ parent\_map$($current$)$
}
\Return{$\mathrm{result}$}\;
\end{algorithm}

\begin{figure}
\centering
\includegraphics[width=\linewidth]{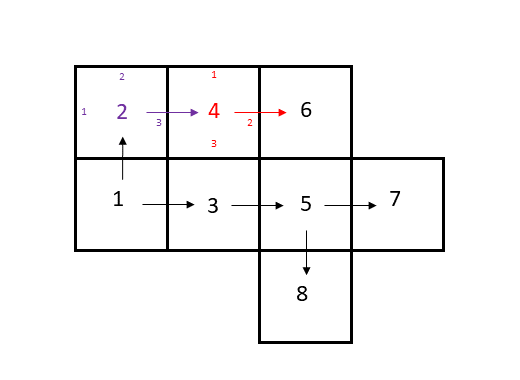}
\caption{  Illustration of Klarner's rules for enumerating squares in a polyomino. Labeling begins at the lowest square in the leftmost column. The next square is selected by considering neighbors of the lowest labeled square, and going clockwise from the face connecting the existing square to the original polyomino (or from the north facing face for the first square; examples in purple and red provided in the figure). A bit description for each tile is then written down, recording tiles which `grow' from that particular tile. For example, tile $2$ has the description $(0,0,1)$, while tile $4$ has the description $(0,1,0)$.}
\label{fig:Klarner}
\end{figure}

The output of Klarner's algorithm is therefore a series of bits of length $3N$ (note the first bit is omitted in Klarner's original argument, leading to $3N-1$ bits. For convenience, this bit exists in our implementation but is always `0'). Importantly, every polyomino is assigned a unique bit description, and there exists an inverse algorithm to recover a given fixed polyomino from Klarner's bit description, and the length of this program does not scale with $N$ (Algorithm \ref{alg:KlarnerInverse}). Hence, the scaling of the worst case Kolmogorov Complexity of a an algorithm for consturcting any given shape of size $N$ is upper bounded by $O(N)$ (note that adding equivalence between rotations at worst adds a constant to the program length). 

Once again, these bounds do not preclude a universal assembly kit, this time with $M \geq 5$. For some finite $M\geq 5$ there are more sequences of length $N$ than shapes of size $N$ in the shape space; for some finite $M$ there are more bits stored in a sequence of length $N$ than bits required to write a program that constructs any shape of size $N$ in the shape space. Despite these facts, however, as Theorem \ref{theorem:existence} shows, no universal assembly kit exists for the sequenced aTAM. The inefficiency of the assembly rules are so strong that, not only is the number of tiles far larger than $M=5$, no finite $M$ is sufficient at all.

%Our analysis of Kolmogorov complexity does not reveal any fundamental algorithmic reason that the shape space of sequenced aTAM should be unencodable with a sequence of tiles drawn from a finite set. 

Klarner's algorithm itself gives us some insight into the cause of this inefficiency. If one were able to use a queue, as in Klarner's algorithm, so that new tiles could only be added adjacent to the tile in a configuration that was at the top of a queue, it would be straightforward to identify a universal assembly kit. This queue-based geometric constraint would effectively be the generalization of the backbone to a process that allows for branching points, rather than requiring a single self-avoiding walk. However, implementing such a queuing system in a molecular asssembly process rather than {\it in silico} would be highly non-trivial.

\begin{algorithm}
\caption{Klarner's inverse algorithm for obtaining the coordinates of a given polyomino from its bit description.}\label{alg:KlarnerInverse}
\KwIn{array of bit triples tripleList}
$N \leftarrow$ size$($tripleList$)$\;
$S \leftarrow$ array$(\{(0,0)\})$\;
parentMap $\leftarrow \mathrm{Map}\{\mathrm{current'} \to \mathrm{parent'}\}$\;
parentMap$((0,0)) \leftarrow (0,-1)$\;
\For{$i \leftarrow 1$ \KwTo $N$}{
    current $\leftarrow  S[i]$ \;
    parent $\leftarrow $ parentMap$($current$)$ \;
    \For{$j \leftarrow 1$ \KwTo $3$}{
        \If{$\mathrm{tripleList}[i][j]$}{
            neighbor $ \leftarrow$ neighbor $j$ of current,  clockwise with parent at $j = 0$\;
            parentMap(neighbor) $\leftarrow$ current \;
            $S$.pushBack(neighbor)
        }
    }
}
\Return{$S$}\;
\end{algorithm}

\subsection{Applying the sequenced aTAM to the shape space of the backboned aTAM}
\label{sec:hamiltonian_sequenced}

Geometric restrictions in tile placement therefore appear to convey some {\it advantages} in allowing the assembly of a complex shape with a small number of tiles.  For a concrete example, in addition to the discussion in Section \ref{sec:Kolmogorov}, consider the configuration shown in Figure~\ref{fig:universal_failure}\,(a); the sequence of tiles shown and interaction function in Theorem \ref{theorem:interacting_walk} will deterministically produce the given shape with the backboned aTAM, starting from an empty configuration. By contrast, using the same model inputs $( A_{empty}, Q, g)$ but employing the placement rules of the sequenced aTAM, assembly is not deterministic. The shape in Figure~\ref{fig:universal_failure}\,(b) can form if the first blue tile binds to a face on the far side of the growing assembly.
 
\begin{figure}
\centering
\begin{subfigure}{0.45\textwidth}
\caption{}
\includegraphics[width=\linewidth]{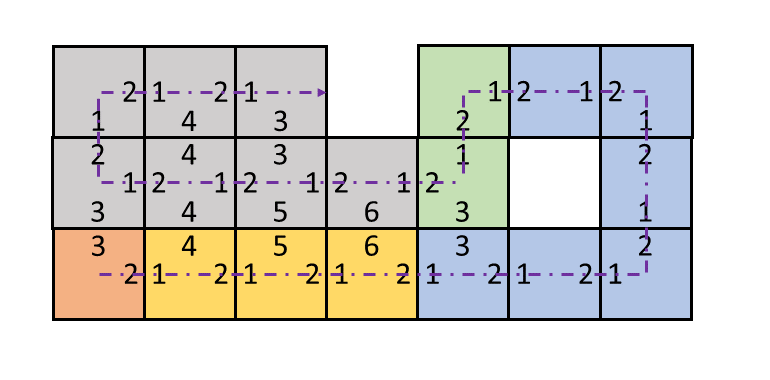}
\end{subfigure}
\begin{subfigure}{0.45\textwidth}
\caption{}
\includegraphics[width=\linewidth]{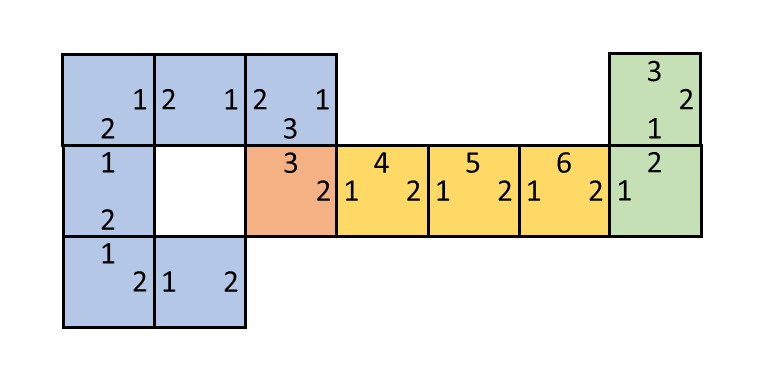}
\end{subfigure}
\caption{ Geometric constraints due to the backbone can improve determinism. \textbf{(a)} The sequence shown (starting with the orange tile) will deterministically produce the shape shown in the backboned aTAM, given an empty initial configuration and the interaction function in Equation~\ref{eq:g_2}. \textbf{(b)} The sequenced aTAM , using the same initial configuration, interaction function and sequence can produce both (a) and (b), and hence is not deterministic. }
\label{fig:universal_failure}
\end{figure}

 If the geometrical restrictions of the backbone appear to provide an advantage, it is natural to ask whether a finite assembly kit exists that allows the sequenced aTAM to deterministically assemble any shape in the shape space of the backboned aTAM. If we allow for non-attractive inter-tile interactions, identifying such a kit is trivial, as our construction in Theorem \ref{theorem:walk_main} would still satisfy the sequenced aTAM. However, as we have argued, assembly with only attractive inter-tile  interactions is conceptually more interesting. Clearly, the construction in Theorem \ref{theorem:interacting_walk} does not work for the sequenced aTAM. When attempting to produce rectangles of $m$ rows and $n$ columns (as in Section \ref{section: individual_kits}) with the sequenced aTAM $\mathcal{P}_{m,n} = ( A_{empty}, Q_{m,n}, g)$ using a sequence $Q_{m,n}$ and interaction function $g$ derived from the constructive proof of Theorem \ref{theorem:interacting_walk}, $\mathcal{P}_{m,n}$ fails to deterministically produce the desired rectangle for any $n > 2$ (except when $(m,n) = (2,3)$). However, we have thus far been unable to prove or disprove the existence of a finite tile type set for the sequenced aTAM that can deterministically assemble any Hamiltonian path shape with only attractive inter-tile interactions. We now present a partial result towards such a proof. 

We propose a scheme for assembling any enclosed shape in the sequenced aTAM without relying on non-attractive inter-tile interactions. We define an \textbf{enclosed shape} as any shape with a well defined Hamiltonian-cycle boundary, that additionally also contains all points within this boundary. Additionally, we also require that the Hamiltonian cycle boundary has no coordinate that neighbours more than two coordinates within the boundary.

\begin{figure}
   \centering
    \begin{subfigure}{0.4\textwidth}
\centering
  \caption{\label{fig:rectangle_fill}}
  \includegraphics[width=\linewidth]{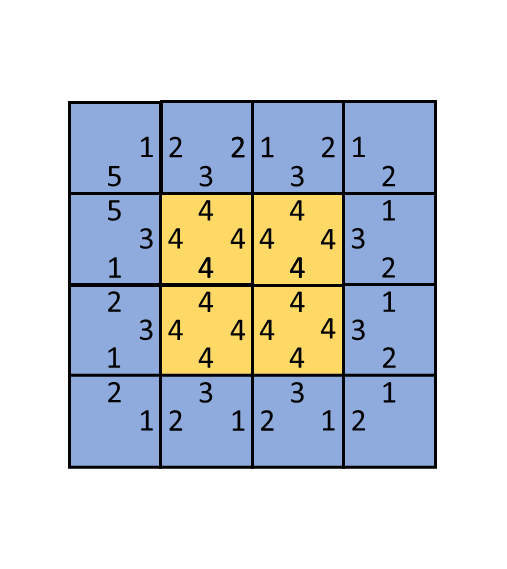}
\end{subfigure}
\begin{subfigure}{0.4\textwidth}
\centering
  \caption{\label{fig:bulge_rectangle_fill}}
  \includegraphics[width=\linewidth]{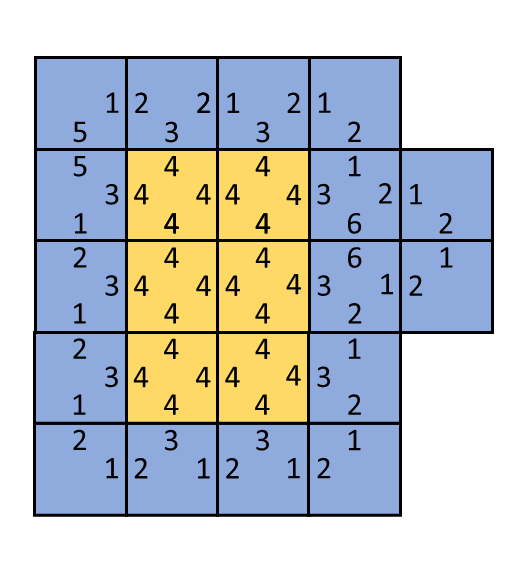}
\end{subfigure}
    \caption{ Scheme for assigning a sequence of tile types to assemble \textbf{(a)} rectangles and \textbf{(b)} bulged rectangles, in the sequenced aTAM with only attractive inter-tile interactions. First, an outer boundary is assembled using interacting directed tiles \textbf{(blue)}. Then, the interior of the shape is filled \textbf{(yellow)}. Note the bulged rectangle itself is not an enclosed shape, but the same shape filling approach can be applied after adjusting for faces in the bulge. }
    \label{fig:fill_shape}
\end{figure}

\begin{theorem}
\label{theorem:}
There exists a finite tile type set with associated interaction function $g$ such that, for any enclosed shape, a sequenced aTAM instance $(A_{empty},Q,g)$ with elements of $Q$ drawn from this tile set that deterministically assembles this shape with only attractive inter-tile interactions.
\label{theorem:enclosed_shape}
\end{theorem}
\begin{proof} Let $g$, with domain $[0,5] \times [0,5]$ and $g(\sigma,\sigma') = g(\sigma',\sigma)$, be defined as follows:
\begin{equation}
    g(\sigma,\sigma')=
    \left\{
    \begin{array}{ll}
      1, & \mbox{if}\ (\sigma,\sigma') = (1,2) \mbox{, } (\sigma,\sigma') = (3,4) \mbox{, } (\sigma,\sigma') = (4,4) \mbox{ or } (\sigma,\sigma') = (5,5)\\
      -3, & \mbox{otherwise}.
    \end{array}
    \right.
\label{eq:g_2_enclosed}
\end{equation}
Then, consider a the subset of the interacting directed tiles from Theorem \ref{theorem:interacting_walk} that only contain $3$ as a non-backbone glue type. Any enclosed shape can be constructed by drawing the boundary with interacting directed tiles, with the non-backbone glue 3 facing inwards towards the shape interior. Finally, the interaction between the first and last added boundary tile is encoded by $(5,5)$. Then, the remainder of the sequence $Q$ is filled with the tile type $\{4,4,4,4\}$ to fill in the shape (For an example, refer to Figure \ref{fig:rectangle_fill}).\end{proof}

\begin{figure}[t]
\centering
\begin{subfigure}{0.49\textwidth}
\centering
  \caption{\label{fig:rectangle_seq}}
  \includegraphics[width=\linewidth]{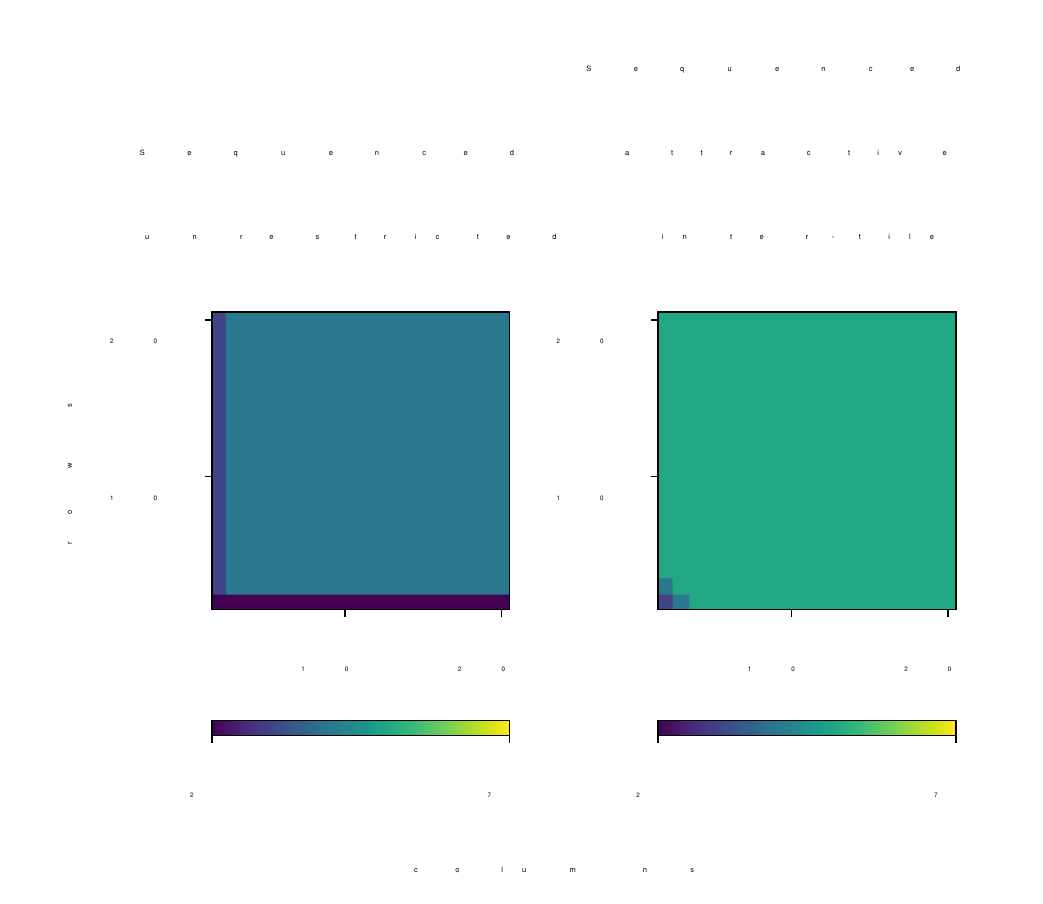}
\end{subfigure}
\begin{subfigure}{0.49\textwidth}
\centering
  \caption{\label{fig:bulge_rectangle_seq}}
  \includegraphics[width=\linewidth]{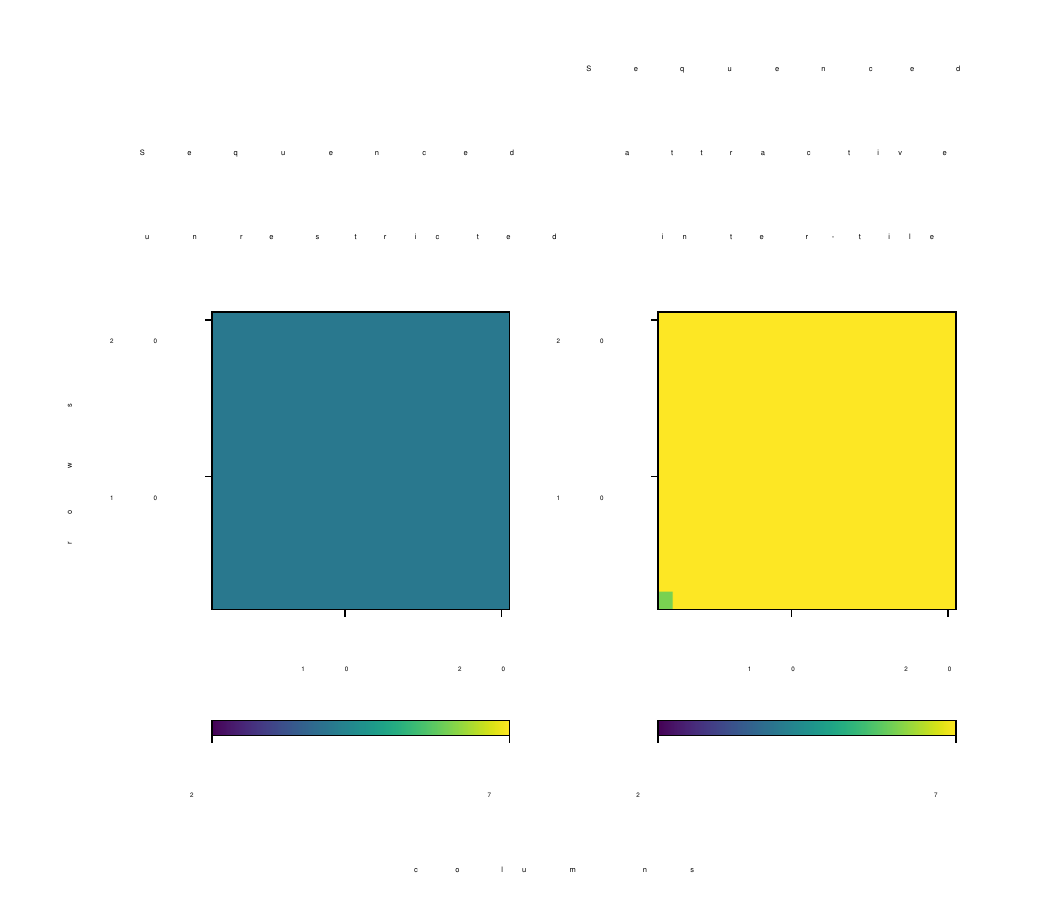}
\end{subfigure}
\caption{ Upper bounds on the tile complexity, or the number of tiles needed to construct \textbf{(a)} rectangles and \textbf{(b)} bulged rectangles of various row and column numbers using sequenced aTAM, with and without interaction restrictions. There is an increase in the number of tiles needed to assemble both shapes when inter-tile interactions are forced to be attractive, but in general this increase is small. }
\label{fig:individual_assembly_kit_sequenced}
\end{figure}

We now consider upper bounds on the number of tiles required to produce the rectangle and bulged rectangle shapes in Figures \ref{fig:rectangle_seq} and \ref{fig:bulge_rectangle_seq} with unconstrained Sequenced aTAM and Sequenced aTAM with only attractive inter-tile interactions. With the former, it is possible to assemble these shapes in the same way as backboned aTAM, using left-right-forward tiles to follow a hamiltonian path. With the latter, we apply the shape filling approach in Theorem \ref{theorem:enclosed_shape}, making slight adjustments for the bulge in the bulged rectangle (note the bulged rectangle itself is not an enclosed shape, but the shape filling approach can be applied after this adjustment). Results are provided in Figure \ref{fig:individual_assembly_kit_sequenced}. For shapes where shape filling is possible, it is quite efficient, only requiring a few more tiles than unconstrained sequenced aTAM. 

Moreover, we end by noting that, just as tile sequences that deterministically assemble a shape in the backboned aTAM may fail if used in the corresponding sequenced aTAM instance, sequenced aTAM instances cannot generally be directly converted into a backboned aTAM instance that deterministically assembles a shape. For instance, using the scheme in Theorem \ref{theorem:enclosed_shape} may fail for a backboned aTAM instance. For a large enough shape, the trajectory may take an arbitrary route inside the interior yellow portion of the shape instead of filling it completely.  

\section{Conclusion}\label{sec:conclusion}

That sequence-directed assembly can help reduce the number of tiles required to assemble shapes is intuitively understood. Here, we show that the details of a sequence-directed assembly map have important implications on the ability of the given assembly map to assemble all shapes in its shape space with a finite set of building blocks. Through our results on the sequenced aTAM, we showed that keeping the sequence encoding of tiles but removing geometric restrictions on the positions of added tiles allows a larger space of possible shapes, but excludes the possibility of a universal assembly kit for these shapes. Our results tentatively suggest that the backbone constraint not only restricts the space of shapes that can be assembled, but also actively facilitates deterministic assembly with a small tile set, since ambiguities of where to place tiles can be avoided. So far, our results have explored rotatable but not flippable tiles, and it would be interesting to consider whether allowing tiles to flip would have any significant impacts on our results. Alternatively, taking the model into 3 dimensions would, presumably, inhibit the effectiveness of the shape-filling approach for the sequenced aTAM and may result in less contrived structures than treeangles that cannot be assembled with a finite assembly kit. Another key question that remains open is whether the sequenced aTAM can deterministically assemble any shape accessible to the backboned aTAM using a finite assembly kit (under the assumption of attractive inter-tile interactions), and if so, how the minimal complexity of this finite assembly kit compares to that of the backboned aTAM.

An initially surprising result was the degree of inefficiency  with which the information in the sequence program is used to direct assembly in the models with only attractive inter-tile interactions. The mechanistic nature of the assembly process has a large effect, and the failure of many sequences to form deterministic, well-defined structures necessitates more tile types, which in turn results in redundancy, further increasing inefficiency.

This highly abstract study was initially motivated by the biologically relevant question of what a pre-formed backbone contributes to the self-assembly of RNA and proteins. Our results suggest that, at least in simplified models, the presence of a backbone has a qualitative, rather than merely quantitative, effect on the number of structures that can be reliably formed with a finite set of building blocks, and that both the sequence and the geometric constraints imposed by the backbone are important in directing successful assembly. These results are likely to generalise beyond the specific details of the model considered here. Moreover, it is interesting to note that both proteins \cite{Li1996} and RNA \cite{Dingle2015} have large neutral spaces, within which many sequences fold into the same structure, and also that many sequences do not produce well-defined folds \cite{Sidl2025}. Both imply some degree of inefficiency in terms of converting the sequence program into structure formation.

Nonetheless, it would also be instructive to explore more realistic models. In particular, biological assembly occurs at finite temperature with interactions of moderate strength. Contacts can form and break, and cooperativity between units is essential in forming long-lasting bonds. Indeed, the need for cooperative interactions may explain how RNA folding can operate with only 4 unique bases, and protein folding can occur with only 20 different side chains; in practice, cooperatively interacting domains likely provide a much larger effective set of building blocks to direct assembly. Testing this hypothesis with a model of backboned assembly that incorporates finite interaction strengths and cooperative bonding would be a natural next step.

%\backmatter

\section*{Ethical Approval} Not applicable, as no human or animal subjects were involved in this work.

\section*{Competing Interests} All authors certify that they have no affiliations with or involvement in any organization or entity with any financial interest or non-financial interest in the subject matter or materials discussed in this manuscript.

\section*{Funding} JG was supported by an Imperial College President’s PhD Scholarship, and TEO by a Royal Society University Research Fellowship.

\section*{Authors' Contributions} 
JG and TEO planned the research. JG performed the research. JG and TEO wrote the manuscript.

\section*{Data Access Statement} 
Data and code may be found at 10.5281/zenodo.15020244.

\ack
We thank Ard Louis, Jordan Juritz and Benjamin Qureshi for their helpful discussions.

\section*{References}
\bibliography{main}

\begin{appendices}

\section{Proof that an Unbounded Number of Tile Types are Required for the Sequenced aTAM with Repulsive Interactions}\label{secA1}

Theorem \ref{theorem:InfiniteInteractions} below replaces Theorem \ref{theorem:infinite_interactions} in the main text for the case where repulsive interactions are permitted. Theorem \ref{theorem:InfiniteInteractions} has additional assumptions on the properties of the shape in question relative to Theorem \ref{theorem:infinite_interactions}, but the treeangle shape defined in the main text was constructed to satisfy this extended set of assumptions. Hence, Theorem  \ref{theorem:InfiniteInteractions} leads directly into Lemma \ref{theorem:scale} in the main text.

Now, we extend the arguments from Theorem \ref{theorem:infinite_interactions} for cases where the strength function $g$ is allowed to take on the values $g < 0$. The only difference is we now need to consider tile placement being blocked by repulsion, rather than just overlap of tiles, which could in principle result in the growth of distinct shapes from the same glue type. To understand how this may happen, consider the following set up. Let $\Psi$ be a complete trajectory of $\mathcal{P}$, with entries $A_t = A_{0,t} + A_{1,t} + A_{2,t} + ... + A_{M,t}$ where $E(A_i) = S_i$ and $A_{i,t}$ is the subconfiguration of $A_i$ appearing at time $t$ in $\Psi$. Assuming that two `growth' faces have the same glue $A_0(\ddot{z}_i,\ddot{k}_i) = A_0(\ddot{z}_j,\ddot{k}_j)$, using the approach in Lemma \ref{theorem:the_beginning}, we can construct another `partial' trajectory $\Psi'$ with the subconfigurations $A_{i,t}$ and $A_{j,t}$ swapped that stops at the first time point $t_s$ when repulsion due to an interaction with $g<0$ precludes a tile from being added to an empty coordinate for the first time. As with Theorem \ref{theorem:the_beginning}, we can write $A_{t}' = A_0 +  A_{i \rightarrow j ,t} + A_{j \rightarrow i,t} + A_{c,t}$ as the entries of the trajectory  $\Psi'$. (As a reminder, $A_{j \rightarrow i}$ represents the configuration $A_j$ transformed so that it now `grows' from $(\ddot{z}_i,\ddot{k}_i)$ instead).

For this trajectory $\Psi'$, at some time $t_s$, a single tile configuration $a_{i}^{des}$ with tile consistent with entry $t_s$ of $Q$ forms an attractive interaction with configuration $A_{j \rightarrow i ,t_s}$, but is prevented from being added because it forms some repulsive interaction with another tile. We call $a_i^{des}$ a \textbf{destabilizing tile}, and its corresponding coordinate $z_i^{des}$ a \textbf{destabilizing coordinate}. The trajectory $\Psi'$ can then proceed, and eventually terminate either prematurely or when the last tile of $Q$ is reached, thus $\Psi'$ is a complete trajectory. In this complete trajectory, say $A_{j \rightarrow i,t_s}$ grows into a configuration $A_{j \rightarrow i}'$. As a result of the tile blocking event arising from $a_{i}^{des}$, it is possible that deterministic growth can occur that results in $E(A_{j \rightarrow i}') = E(A_i) \not \simeq E(A_{j})$ and $E(A_{i \rightarrow j}') = E(A_j) \not \simeq E(A_{i})$.
%In the context of the assumptions laid out in Theorem \ref{theorem:infinite_interactions}, if $z_i^{des}$ is not subsequently filled at some time $t > t_s$, then no two $S_i$ and $S_j$ neighbour each other. As a result, no assumptions of Theorem \ref{theorem:infinite_interactions} are broken. 
Repulsion-derived tile blocking can thus cause the reasoning in Theorem \ref{theorem:infinite_interactions} to fail.  

The remainder of this appendix is dedicated to illustrating how we can get around this potential failure mode. We begin with a few additional definitions. The \textbf{distance} between two points $(x_1,y_1)$ and $(x_2,y_2)$ is taken to be $|x_2-x_1|+|y_2-y_1|$. A \textbf{subtree subshape} $s$ of an oriented shape $S$ is a subshape of $S$ with a unique \textbf{root} in $s$ and its neighbouring \textbf{origin} outside of $s$ but in $S$ (Figure \ref{fig:subtree_expl}). The adjacency graph $G_S(s)$ forms a tree and all non-root leaf nodes in $G_S(s)$ are leaf nodes in $G_S(S)$. Additionally, each subtree subshape is uniquely specified by its \textbf{root} and \textbf{origin}. Now, we proceed to tighten our definition of the target shape to ensure that tile blocking doesn't invalidate our line of reasoning.

\begin{definition}
    Consider an oriented shape $S = \bigcup_{i = 0}^M S_i$. $S$ is a \textbf{target shape} with \textbf{starting shape} $S_0$ if it obeys the following assumptions: 
\begin{enumerate}
    \item For any $i,j \in {1,..,M}$, $S_i$ is not connected to any $S_j$ if $ i \neq j$. \label{step:noneighbour}
    \item For any $i \in {1,..,M}$, $S_i$ has exactly one coordinate that neighbours a coordinate in $S_0$, and $S_0$ has exactly one coordinate that neighbours this coordinate. We denote by $\dot{z}_i$ the coordinate in $S_i$ and by $\ddot{z}_i$ the neighbouring coordinate in $S_0$. 
    \label{step:exactly1}
    \item For any $i \in {1,..,M}$, $G_S(S_i)$ is a tree with at least one branching node, and with branching points in $S_i$ labeled $v_{i,n}$. Then, every subtree subshape with origin at some $v_{i,n}$ has a unique shape (Figure \ref{fig:subtree_expl}). \label{step: subtree}
    \item For any $i \in {1,..,M}$, $j \in {1,..,M}$ and $i \neq j$, $S_i$ and $S_j$ have no coordinates  within a distance 2 of each other if $\ddot{z}_i \neq \ddot{z}_j$. If $\ddot{z}_i = \ddot{z}_j$, then $\dot{z}_i$ is exactly distance 2 from $\dot{z}_j$, but no other coordinates in $S_i$ are within distance 2 of any other coordinate in $S_j$ (Figure \ref{fig:assume_4_5}.a). \label{step:s_i_j} 
    \item For any $i \in {1,..,M}$, $S_i$ has no coordinate which is within distance 2 from any coordinate in $S_0$, except for $\dot{z}_i$ or its neighbours. There is exactly one coordinate in $S_0$ at a distance of 2 from $\dot{z}_i$, which must be some neighbour of $\ddot{z}_i$. $\ddot{z}_i$ is the only coordinate in $S_0$ that can be distance 2 from any neighbour of $\dot{z}_i$ within $S_i$ (Figure \ref{fig:assume_4_5}.b). 
    \label{step:s_i_0} 
    \item Every straight line segment in $S$ is of at least length 4. \label{step:straight}
    \item $S$ and $S_0$ have rotational symmetry 1. Additionally, there does not exist $S_t$, a subshape of $S$, such that $S_t \neq S_0$ but $S_t \simeq S_0$. \label{step:strong_weak}
    \label{def:targetshape}
\end{enumerate}
\end{definition}

Having expanded our assumptions on our target shape, we wish to consider when repulsion-induced blocking may or may not be significant for our arguments. Consider a sequenced aTAM instance $\mathcal{P} = (A_0, Q, g)$ that produces a target shape $S$ from starting shape that $S_0 = E(A_0)$. The terminal configuration $A_{\zeta}$ of a  trajectory of $\mathcal{P}$ has subconfigurations $A_i$ such that $E(A_i) = S_i$. Assume that partial sub-configurations $A_{j,t_s}$ and $A_{i,t_s}= A_{j \rightarrow i,t_s}$ could be reached during complete trajectories at some time $t_s$. We say that $j$ and $i$ are \textbf{differentially blocked} if, for some single tile configuration $a_j$, $A_{j,t_s}+a_{j}$ can be reached during a complete trajectory but $A_{i,t_s}+a_{j \rightarrow i}$ cannot (or vice versa), due to repulsion-induced blocking
%For a given complete trajectory $\Psi'$, assuming a repulsion-induced tile blocking event does not occur before time $t_s$, then consider a pair of subconfigurations $A_{j,t_s}$ and $A_{i,t_s} = A_{j \rightarrow i,t_s}$. We say that $j$ and $i$ are \textbf{differentially blocked} if, for some single tile configuration $a_j$, $A_{j,t_s}+a_{j}$ can be reached during a complete trajectory but $A_{i,t_s}+a_{j \rightarrow i}$ cannot (or vice versa) due to the presence of negative interactions 
(Hence $a_i^{des} = a_{j \rightarrow i}$ is the destabilizing tile with coordinate $z_i^{des}$). In the following lemma, we show that differential blocking is essential to break the arguments in Theorem $\ref{theorem:infinite_interactions}$.

\begin{lemma}
   Consider a sequenced aTAM instance $\mathcal{P} = (A_0, Q, g)$ that oriented-deterministically produces a target shape $S$ (obeying assumptions in Definition \ref{def:targetshape}) from starting shape that $S_0 = E(A_0)$.  In the absence of differential blocking between two subconfigurations anchored to the same glue type, the number of unique tile types required in $A_0$ grows with at least $\frac{M}{4}$.
   \label{theorem:need_differential}
\end{lemma}
\begin{proof}  Consider whether it is possible to detemrinistically produce an assembly with $S_i \not \simeq S_j$ if $S_i$ and $S_j$ are anchored to the same glue type.
First, if the addition of a tile to partial subconfigurations corresponding to $i$ and $j$ is never blocked due to some repulsive interaction $g < 0$, then Theorem \ref{theorem:infinite_interactions} is sufficient to  argue that $S_i \simeq S_j$ necessarily, since the assumptions of Definition \ref{def:targetshape} are a subset of those of Theorem \ref{theorem:InfiniteInteractions} and the contradictions that arise in Theorem \ref{theorem:InfiniteInteractions} remain upheld. 

We now allow for non-differential repulsion-induced tile blocking. Consider (without loss of generality) a non-differential repulsion-induced tile block that occurs when trying to add a tile to a partial configuration $A_{i,t_s}$, where partial subconfiguration $A_{j,t_s}= A_{i \rightarrow j,t_s}$ can also be reached in a complete trajectory.
%neither $A_{j,t_s}+a_{j}$ nor $A_{i,t_s}+a_{j \rightarrow i}$ can be a subconfiguration of any configuration in the complete trajectory $\Psi$. Hence, if assembly is deterministic and $A_i$ and $A_j$ are anchored to the same glue, such a 
This non-differential tile block does not result in  partial subconfigurations $A_{j,t_s^\prime}$ whose equivalents $A_{i,t_s^\prime} = A_{j\rightarrow i,t_s^\prime}$ cannot be reached in a complete trajectory. Unless differential tile blocking occurs elsewhere, the assumed properties of the target shape then imply that for all partial subconfigurations of $A_j$ that can be reached in a complete trajectory, the equivalent partial subconfiguration of $A_i$ can also be reached and the proof of Theorem \ref{theorem:InfiniteInteractions} still applies. Either $S_i \simeq S_j$ and a violation of assumption \ref{step: subtree} of our target shape arises, or growth must be non-deterministic.

Hence, in the absence of differential blocking, each $S_i$ must be anchored to a unique glue type, so at least $M$ glue types are required in $A_0$. Since each tile type can contain at most $4$ unique glue types, at least $\frac{M}{4}$ tile types, are needed in $A_0$ to assemble $S$ deterministically.
\end{proof}

Lemma \ref{theorem:need_differential} has the following corollary, which states that we would still need $M/4$ tile types if tile blocking only occurs within given subshapes $S_i$ rather than in between them. 

\begin{corollary}
   Consider a sequenced aTAM instance $\mathcal{P} = (A_0, Q, g)$ that oriented-deterministically produces a target shape $S$ (obeying assumptions in Definition \ref{def:targetshape}) from starting shape that $S_0 = E(A_0)$. If the only tile blocking events are such that each destabilizing coordinate $z_{i}^{des}$ only neighbours one $S_i$, the number of unique tile types required in $A_0$ grows with at least $\frac{M}{4}$.
   \label{theorem:need_junction}
\end{corollary}
\begin{proof}
     We once again assume that $S_i$ and $S_j$ are anchored to the same glue type. Consider (without loss of generality) a tile blocking event in which $z_{i}^{des}$ only neighbours $S_i$, and that the tile is being added to the partial subconfiguration $A_{i,t_s}$. Unless a previous differential blocking event has occurred, from the properties of the target shape, the equivalent partial subconfiguration $A_{j,t_s}= A_{i \rightarrow j,t_s}$ can also be reached in a complete trajectory if assembly is to be oriented-deterministic. $z_{i\rightarrow j}^{des}$ would therefore also experience repulsion-induced blocking, and hence blocking events in which the destabilizing coordinate $z_{i}^{des}$ only neighbours one $S_i$ are non-differential. This Corollary then follows directly as a result of Lemma \ref{theorem:need_differential}.
    
%    If, however, a differential tile block occurs but the destabilizing tile $z_i^{des}$ is never subsequently filled, then oriented-determinism cannot hold because using the reasoning of \ref{theorem:infinite_interactions}, $A_{j \rightarrow i}$ can grow in place of $A_i$ and $A_{j \rightarrow i} \not \simeq A_i$ if $A_i$ and $A_j$ are anchored to the same glue. However, if $z_i^{des}$ is subsequently filled, then this tile block cannot prevent $E(A_i) \simeq E(A_j)$, hence breaking assumption \ref{step: subtree} of our target shape. Thus, producing $S$ deterministically still requires $\frac{M}{4}$ tile types in this case, as in Lemma \ref{theorem:need_differential}.
\end{proof}

\begin{lemma} Consider a sequenced aTAM instance $\mathcal{P} = (A_0, Q, g)$ that oriented-deterministically produces a target shape $S$ (obeying assumptions in Definition \ref{def:targetshape}) from starting shape that $S_0 = E(A_0)$. A differential blocking event with  destabilizing coordinate $z_{i}^{des}$ that neighbours $S_i$ and $S_k$ with $i \neq k$ cannot occur. 
\label{theorem:target_unachievable}
\end{lemma}
\begin{proof} We proceed by assuming that there is such a differential blocking event. Then, we show that one of the assumptions on $S$ laid out in Definition \ref{def:targetshape} necessarily is necessarily violated as a result of this differential blocking event. 

First, assume such a tile blocking event occurs, but $z_i^{des}$ is subsequently filled in at some later time (in the trajectory $\Psi'$). Then $S_i$ would neighbour $S_k$, and hence a contradiction arises with either assumption \ref{step:noneighbour} or assumption \ref{step:exactly1}. Hence, we only need to consider differential tile blocks where here $z_i^{des}$ is adjacent to $S_i$ as well as some $S_k$ with $i \neq k$, and $z_i^{des}$ remains unoccupied
in $A_{\zeta}'$, the terminal configuration of $\Psi'$.

We now proceed to argue that in this final case, features of $A_{\zeta}'$ result in a final shape $S$ that necessarily breaks one of the assumptions \ref{step:exactly1}-\ref{step:straight} of Definition \ref{def:targetshape}. $a^{des}_{i}$ neighbours some tile $a^{att}_{j \rightarrow i}$ in $A_{j \rightarrow i, t_s}$, onto which it may be attracted, and some other tile $a^{rpl}$ that blocks its addition. Hence, $a^{att}_{j \rightarrow i}$ and $a^{rpl}$ must be of distance $2$ from each other. The assumptions on $S$ that we have laid out forbid points belonging to any two distinct $S_i$ and $S_j$ from being within distance 2 of each other, with exceptions around the vicinity of junctions where $S_0$ meets some $S_i$. More specifically, in order to fulfill assumptions \ref{step:s_i_j} and \ref{step:s_i_0}, $a^{att}_{j \rightarrow i}$ can only occupy the position $\dot{z}_i$ or some neighbour of $\dot{z}_i$. 

Next we show that by occupying one of these positions, one or more of the assumptions regarding $S$ must be broken. We define some coordinate $z_i^{nb}$ that forms a straight line with $\dot{z}_i$ and $\ddot{z}_i$ and neighbours $\dot{z}_i$. Consider now the following cases, which exhaust the possibilities of differential blocking. We will assume, without loss of generality, that any differential blocking occurs due to a tile being blocked from addition to the subshape $S_i$.
%(note that since we have established that a block in $S_i$ accompanied with an analogous block in $S_j$ cannot help us achieve $S_i \not \simeq S_j$, we explicitly consider here only those cases where a destabilizing coordinate $z_i^{des}$ is not accompanied by an analogous destabilizing coordinate $z_j^{des}$). 

\begin{enumerate}
\item Assume $z_j^{nb} \in S_j$. By assumption \ref{step:s_i_0} and lattice placement rules (as illustrated in Figure \ref{fig:edge_cases_2}.a), there is no way for $z_i^{nb}$ to neighbour some other $S_k$ with $k \neq i$, and hence $z_i^{nb} \neq z_i^{des}$. Hence, there is no way to block an incoming tile from occupying $z_i^{nb}$, and given $S_i$ and $S_j$ are anchored onto the same starting glue types, $z_i^{nb} \in S_i$ if growth is deterministic. Consider then $s_{i,nb}$, the subtree subshape of $S_i$ rooted at $z_i^{nb}$ with origin at $\dot{z}_i$. $s_{i,nb}$ cannot include neighbours of $\dot{z}_i$ other than the root, otherwise $S_i$ as a whole would not be a tree. Hence, all positions in $s_{i,nb}$ other than $z_i^{nb}$ are at a distance greater than 2 from any other $S_k$ for $k \neq i$ (noting that in assumptions \ref{step:s_i_j} and \ref{step:s_i_0}, exceptions only cover $\dot{z}_i$ and its neighbours, so neighbours of $z_i^{nb}$, their neighbours, and so on cannot invoke either exception). So, no position in $s_{i,nb}$ can neighbour any $z_i^{des}$ that neighbours some $S_k \neq S_i$ (incoming tiles onto positions in and adjacent to $s_{i,nb}$ cannot be blocked except by tiles in $s_{i,nb}$).  Hence $s_{i,nb} \simeq s_{j,nb}$ if growth is deterministic. Then, if  $\dot{z}_i$ and $\dot{z}_j$ are branching points, assumption \ref{step: subtree} is violated. Otherwise, as assumption \ref{step: subtree} also requires that each $S_i$ contains a branching point, $s_{i,nb}$ and $s_{j,nb}$ also contain branching points, and so some rooted subtree $s_{i,nb}' \subseteq s_{i,nb} \subseteq S_i$ and $s_{j,nb}'  \subseteq s_{j,nb} \subseteq S_j$ are such that $s_{i,nb}' \simeq s_{j,nb}'$, also violating assumption \ref{step: subtree}.
     \item Assume $z_j^{nb} \notin S_j$. For the reasons above, $z_j^{nb}$ cannot be the position of a blocked tile, so $z_i^{nb} \notin S_i$. If $\ddot{z}_i$ is a branching point, the segment $\{ \ddot{z}_i, \dot{z}_i \}$ must bend onto a third coordinate that is not $z_i^{nb}$, resulting in $\{ \ddot{z}_i, \dot{z}_i \}$ being a straight line segment of length 2 (Figure \ref{fig:edge_cases_2}.b), contradicting assumption \ref{step:straight}. If $\ddot{z}_i$ is not a branching point, then by assumption \ref{step:s_i_j}, $z_i^{des}$ must neighbour $S_0$ (since it cannot neighbour any other $S_k$ for $k \neq i$). $z_i^{des}$ also neighbours either $\dot{z}_i$ or one of its neighbours. Assume first it neighbours $\dot{z}_i$. The position that neighbours $z_i^{des}$ in $S_0$ must be within distance 2 of $\dot{z}_i$, so this position necessarily bends onto the segment $\{\ddot{z}_i,\dot{z}_i\}$. This segment, in turn, bends onto a third coordinate that is not $z_i^{nb}$, and hence $\{\ddot{z}_i,\dot{z}_i\}$ is once again a straight line segment of length 2 (Figure \ref{fig:edge_cases_2}.b), contradicting assumption \ref{step:straight}. Finally, if $z_i^{des}$ neighbours a neighbour of $\dot{z}_i$, it must neighbour $\ddot{z}_i$ as well by assumption \ref{step:s_i_0}. Similarly, the equivalent position $z_{i \rightarrow j}^{des}$ must neighbour $S_0$. 
     %This means that, if an analogous position $z_{i \rightarrow j}^{des}$ was occupied in $S_j$, then both $z_{i \rightarrow j}^{des}$ and $\dot{z}_j$ would neighbour $S_0$. 
     Thus either blocking is non-differential,  or the structure violates assumption \ref{step:exactly1}.
     %%{\color{red}[The point is that either this location is also blocked in $S_j$ and thus $S_j \simeq S_i$, OR $S_j$ is illegal? You haven't quite set the whole thing up as ``I must be able to place a tile in $j$ that I can't when switched to $i$ if I'm going to get $S_i \not \simeq S_j$'', but that is basically what you seem to be using. Alternatively, something like ``a block that occurs in both $i$ and $j$ doesn't help".    Fundamentally, the argument in "1" is you can have a differential block, but it doesn't stop you growing an identical subtree, whereas the argument in "2" is that you can't have a differential block and have a valid structure.]}
\end{enumerate}

Hence, a contradiction in one of our assumptions arises if a differential block with $z_i^{des}$ that neighbours $S_i$ and $S_k$ with $k \neq i$ occurs.
\end{proof}

\begin{theorem}
Consider a sequenced aTAM instance $\mathcal{P} = (A_0, Q, g)$ that deterministically produces a target shape $S$ (obeying assumptions in Definition \ref{def:targetshape}) from a starting shape $S_0 = E(A_0)$. The number of unique tile types required in $A_0$ grows with at least $\frac{M}{4}$.
\label{theorem:InfiniteInteractions}
\end{theorem}
\begin{proof}
Assumption \ref{step:strong_weak} and determinism imply that $\mathcal{P}$ is oriented-deterministic. As per Lemma \ref{theorem:need_differential} and Corollary \ref{theorem:need_junction}, only differential blocking where $z_i^{des}$ neighbours $S_i$ and $S_k$ with $k \neq i$ will allow $A_0$ to require fewer than $\frac{M}{4}$ tile types. However, Lemma \ref{theorem:target_unachievable} shows that such a differential tile block necessarily results in $S$ failing one of assumptions \ref{step:exactly1}- \ref{step:straight} in Definition \ref{def:targetshape}. Hence, a target shape fulfilling the assumptions of Definition \ref{def:targetshape} cannot be deterministically assembled by $\mathcal{P}$ if $A_0$ has fewer than $\frac{M}{4}$ tile types. 
\end{proof}

The treeangle shape defined in the main text was constructed such that it can be partitioned in a way that satisfies the assumptions of Definition \ref{def:targetshape}. Hence this Theorem can be used in the proof of Lemma \ref{theorem:scale}, taking the place of Theorem \ref{theorem:infinite_interactions} which is only valid if $g \geq 0$.

\begin{figure}
 \includegraphics[width=\linewidth]{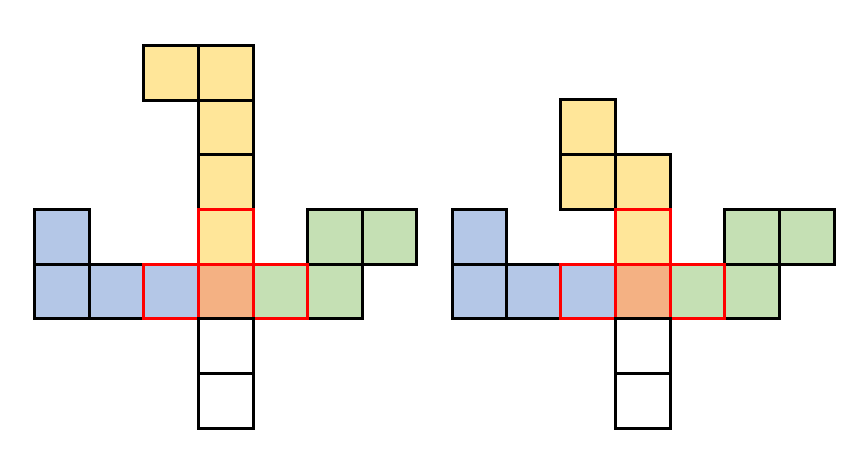}
  \caption{An illustration of assumption \ref{step: subtree} in the definition of shapes $S_i$ used in Theorem \ref{theorem:InfiniteInteractions}. The two shapes are two instances of $S_i$, where the lower white tiles connect to some larger $S_0$. A single branching point (also the \textbf{origin tile} of the subtree subshapes) is given in orange, while \textbf{root tiles} are outlined in red. The left oriented shape obeys assumption \ref{step: subtree}, as each subtree subshape (coloured in blue, yellow and green) are distinct, while the right oriented shape violates this assumption as the yellow rooted subtree subshape is equivalent to the green under a $90 ^{\circ}$ clockwise rotation followed by a $(1,-1)$ translation. }
  \label{fig:subtree_expl}
\end{figure}

\begin{figure}
\begin{subfigure}[b]{0.45\textwidth}
    \caption{}
      \includegraphics[width=\textwidth]{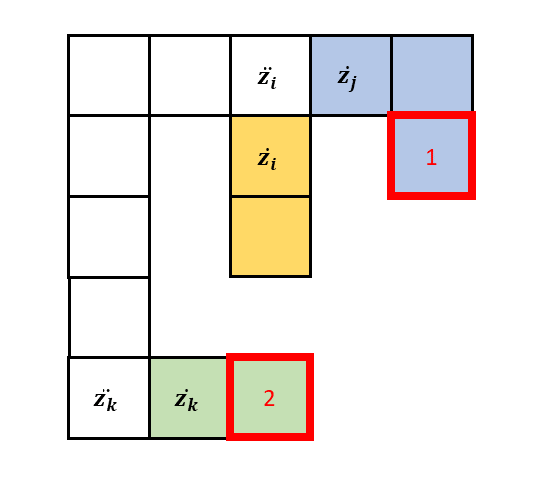}
    \end{subfigure}
    \begin{subfigure}[b]{0.45\textwidth}
    \caption{}
      \includegraphics[width=\textwidth]{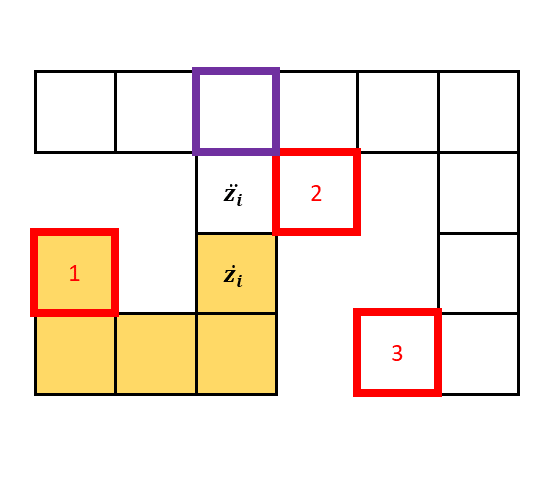}
    \end{subfigure}
 
  \caption{Illustrations of assumptions \ref{step:s_i_j} and \ref{step:s_i_0} in  the definition of shapes in Theorem \ref{theorem:InfiniteInteractions}. Tiles belonging to $S_0$ are in white, and tiles belonging to each $S_i$ are assigned a single color. Tiles violating assumptions are outlined in red. \textbf{a.} An illustration of how assumption \ref{step:s_i_j} can be violated. Violation 1 is due to a tile of $S_j$ (with $\ddot{z}_j = \ddot{z}_i$) that is not $\dot{z}_j$ being within distance $2$ of $\dot{z}_i$, while violation 2 is due to a tile of $S_k$ ($\ddot{z}_k \neq \ddot{z}_i$) being within distance $2$ of an arbitrary tile of $S_i$. \textbf{b.} An illustration of how assumption \ref{step:s_i_0} can be violated. Violation 1 is due to a tile of $S_0$ being within distance 2 of a tile of $S_i$, with both tiles being far from $\ddot{z}_i$ and $\dot{z}_i$. Violation 2 is due to two tiles (the violating tile and the purple outlined tile) in $S_0$ being within distance 2 of $\dot{z}_i$. Finally, violation $3$ is due to a tile of $S_0$, far from $\ddot{z}_i$, being within distance 2 of a neighbour of $\dot{z}_i$. }
  \label{fig:assume_4_5}
\end{figure}

\begin{figure}
\begin{subfigure}[b]{0.45\textwidth}
    \caption{}
      \includegraphics[width=\textwidth]{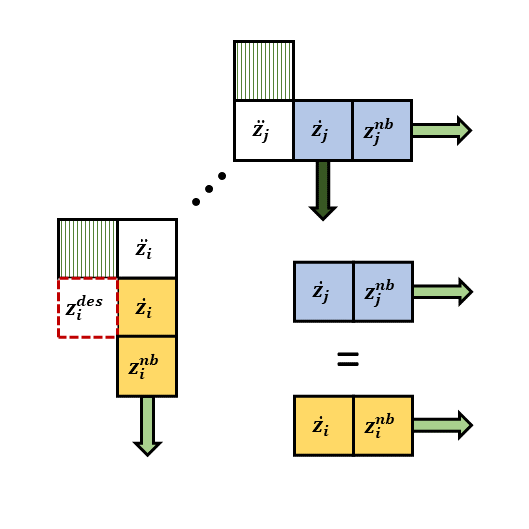}
    \end{subfigure}
    \begin{subfigure}[b]{0.45\textwidth}
    \caption{}
      \includegraphics[width=\textwidth]{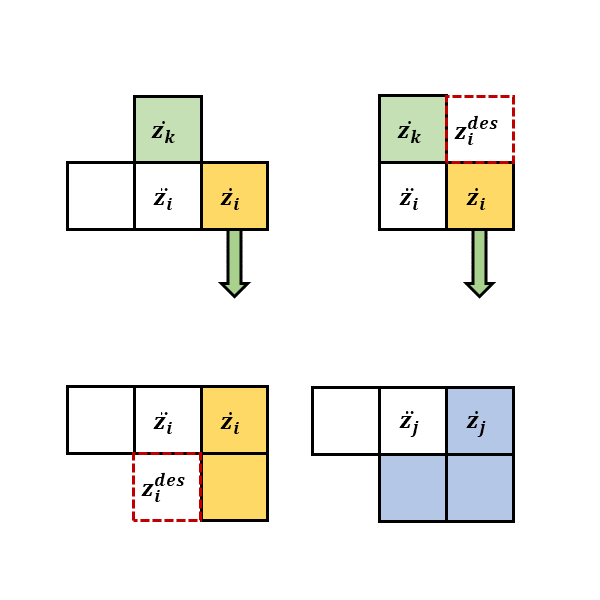}
    \end{subfigure}
 
  \caption{Theorem \ref{theorem:InfiniteInteractions} lays out two possibilities in the case where a tile-blocked coordinate $z^{des}_i$ remains unfilled at the end of a complete trajectory. We show here fragments of the terminal configuration $A_{\zeta,c}'$. $\mathbf{A_{j \rightarrow i}}$ is given in yellow, $\mathbf{A_{j}}$ is given in blue and \textbf{destabilizing tiles} $\mathbf{z_i^{des}}$ are in dotted red,. We illustrate how the two possibilities violate assumptions on $S$ laid out at the beginning of Theorem \ref{theorem:InfiniteInteractions}. \textbf{a.} An illustration of how case 1 violates assumptions on $S$. The striped tile can be either a \textbf{tile of $\mathbf{A_0}$ or $\mathbf{\dot{z}_k}$.}  \textbf{Arrows} represent some arbitrary tree-shaped configurations, where arrows of the same color represent configurations with equivalent (up to rotations and translations) shapes. The tile at $z_j^{nb}$ cannot be destabilized upon transformation into $z_i^{nb}$ as it is too far away from any coordinate of $S_0$ or any other $S_k$, and hence $z_i^{nb} \neq z_i^{des}$. Since branches that grow from $z_i^{nb}$ are too far away from any $S_k$ for $k \neq i$ to be blocked, The branches that grow from $z_i^{nb}$ and $z_j^{nb}$ must have identical shapes, violating assumptions on the shape $S$. \textbf{b.} An illustration of how case 2 violates assumptions on $S$. Green tiles are $\dot{z}_k$ for some $k\neq i,j$. The green arrows represent \textbf{arbitrary tree-shaped configurations}. If $\ddot{z}_i$ is branching (Top left), or if $z_i^{des}$ neighbours $\dot{z}_i$ (Top right), $\{ \ddot{z}_i, \dot{z}_i  \}$ is a straight line segment of length 2, breaking assumption \ref{step:straight}. Otherwise, if the destabilizing tile neighbours $\ddot{z}_i$ and a neighbour of $\dot{z}_i$ (Bottom left), but its equivalent is present in $S_j$, then $S_j$ has two coordinates neighbouring $S_0$ (Bottom right). }
  \label{fig:edge_cases_2}
\end{figure}
\end{appendices}

\end{document}